\documentclass{article}
\usepackage{amssymb}
\usepackage{amsmath}
\usepackage{amsthm}
\usepackage{mathtools}
\usepackage{latexsym}
\usepackage{amstext}
\usepackage{xspace}
\usepackage{enumerate}
\usepackage{macrosLLNCS}
\usepackage{graphicx}
\usepackage{fullpage}
\usepackage{color}

\def\>{\ensuremath{\rangle}}
\def\<{\ensuremath{\langle}}
\def\h{\ensuremath{\mathcal{H}}}

\def\lh{\ensuremath{\mathcal{L(H)}}}
\def\dh{\ensuremath{\mathcal{D(H})}}

\def\r{\ensuremath{\mathcal{R}}}

\def\ra{\ensuremath{\rightarrow}}
\def\e{\ensuremath{\mathcal{E}}}

\def\d{\ensuremath{\mathcal{D}}}
\def\dh{\ensuremath{\mathcal{D(H)}}}
\def\lh{\ensuremath{\mathcal{L(H)}}}

\newcommand {\nil} {\mbox{\bf{nil}}}
\newcommand {\iif} {\mbox{\bf{if}}}
\newcommand {\then} {\mbox{\bf{then}}}
\newcommand {\eelse} {\mbox{\bf{else}}}

\newcommand{\ptr}{{\rm ptr}}
\newcommand{\rto}[1]{\stackrel{#1}\longrightarrow}
\newcommand{\Rto}[1]{\stackrel{#1}\Longrightarrow}

\newcommand{\define}{\stackrel{def}=}

\newcommand{\obis}{\approx_o}
\newcommand{\dist}[1]{\mathop{\mbox{$\mathcal D$}} ({#1})   } % distributions
\newcommand{\pdist}[1]{\overline{#1}  } % distributions
\newcommand{\support}[1]{\lceil{#1}\rceil}
\newcommand{\lift}[1]{\mathrel{{#1}^\dag}}

\newcommand{\Act}{\ensuremath{\mathsf{Act}}\xspace}

\newcommand{\setof}[2]{\{ \, #1 \, \mid \, #2 \, \}}% set comprehension

\newcommand{\pair}[1]{\langle{#1}\rangle}

\newcommand{\so}{{\cal SO}}
\newcommand{\Con}{{\it Con}}
\newcommand{\qv}{{\it qv}}
\newcommand{\qc}{{\tt c}}
\newcommand{\tr}{{\rm tr}}
\newcommand{\CE}{{\cal E}}
\newcommand{\CH}{{\cal H}}
\newcommand{\CL}{{\cal L}}
\newcommand{\CC}{{\cal C}}
\newcommand{\CD}{{\cal D}}
\newcommand{\ifthen}[2]{{\textbf{if} ~#1~ \textbf{then} ~#2}}

\newcommand{\barb}[2]{\mathop{{\Downarrow}_{#1}^{#2}}}

\newcommand{\hyperTS}[1]{\mathrel{\dar{#1}\kern-.4em\blacktriangleright}}

\newcommand{\diam}[1]{\langle#1\rangle}

\newcommand{\rbc}{\approx_{r}}

\newcommand{\obisi}{\approx_{o}}

\newtheorem{theorem}{Theorem}[section]

\newtheorem{lemma}[theorem]{Lemma}
\newtheorem{definition}[theorem]{Definition}
\newtheorem{proposition}[theorem]{Proposition}
\newtheorem{corollary}[theorem]{Corollary}

\begin{document}
\title{Open Bisimulation for Quantum Processes}
%\titlerunning{}
\author{Yuxin Deng$^1$ \qquad Yuan Feng$^2$ \\
$^1$Shanghai Jiao Tong University, China\\
$^2$ University of Technology, Sydney, Australia, and Tsinghua University, China
}

\maketitle

\begin{abstract}
Quantum processes describe concurrent communicating systems that may involve quantum information.
We propose a notion of open bisimulation for quantum processes and show that it provides both a sound and complete proof methodology for a natural extensional behavioural equivalence between quantum processes. We also give a modal characterisation of open bisimulation,
by extending the Hennessy-Milner logic to a quantum setting.
\end{abstract}

%%%%%%%%%%%%%
\section{Introduction}

The theory of quantum computing has attracted considerable research
efforts in the past twenty years. Benefiting from the
superposition of quantum states and linearity of quantum
operations, quantum computing may provide considerable speedup over
its classical analogue \cite{Sh94,Gr96,Gr97}. 
However, functional quantum computers
which can harness this potential in dealing with practical applications  are extremely difficult to
implement. On the other hand,  quantum cryptography, of which the security and ability to detect
the presence of eavesdropping are provable based on the principles of quantum mechanics,
has been developed so rapidly that quantum
cryptographic systems are already commercially available by a number of companies such as Id Quantique, Cerberis, MagiQ Technologies, SmartQuantum, and NEC. 

As is well known, it is very difficult to guarantee the correctness of classical communication protocols
at the design stage, and some simple protocols were finally found to have fundamental flaws.
Since human intuition is poorly adapted to the quantum world, quantum protocol designers will definitely make more faults than classical
protocol designers, especially when more and more complicated quantum protocols can be
implemented by future physical technology.
%With the purpose of cloning the success classical process algebras achieved in analyzing and verifying
%classical communication protocols, various quantum process algebras
%have been proposed by several research groups. 
In view of the success that classical process algebras \cite{ccs,Hoa85,BW90} achieved  in
analyzing and verifying classical communication protocols, several
research groups proposed various quantum process algebras with the
purpose of modeling
quantum protocols.
Jorrand and Lalire \cite{JL04, La06} 
defined a language QPAlg (Quantum
Process Algebra) by adding primitives expressing unitary
transformations and quantum measurements, as well as communications of
quantum states, to a CCS-like classical process algebra. An operational
semantics of QPAlg is given, and further a probabilistic branching
bisimulation between quantum processes is defined.
Gay and Nagarajan \cite{GN05, GN06} proposed a language CQP (Communicating Quantum Processes),
which is obtained from the pi-calculus \cite{MP92} by adding primitives for measurements
and transformations of quantum states, and allowing transmission of qubits.
They presented a type system for CQP,
and in particular proved that the semantics preserves typing and that typing
guarantees that each qubit is owned by a unique process within a system.
The second author of the current paper, together with his colleagues, proposed a language named qCCS \cite{FDJY07, YFDJ09, FDY11} for quantum communicating systems by adding quantum input/output and quantum operation/measurement primitives to classical value-passing CCS \cite{He91,HI93}.
One distinctive feature of qCCS, compared to QPAlg and CQP, is that it provides a framework to describe, as well as reason about, the communication of quantum systems
which are entangled with other systems. Furthermore, a bisimulation for processes in qCCS has been introduced, and the associated bisimilarity is proven  to be a congruence with respect to all process constructors of qCCS. Uniqueness of the solutions
to recursive process equations is also established, which provides
a powerful proof technique for verifying complex quantum protocols.

In the study of quantum systems, as well as classical communicating systems, an important problem
is to tell if two given systems exhibit the same behaviour. To approach the problem we first need to give criteria for reasonable behavioural equivalence. Two systems should only be distinguished on the basis of the chosen criteria. Therefore, these criteria induce an extensional equivalence between systems, $\approx_{\text{\tiny behav}}$, namely the largest equivalence which satisfies them. 

Having an independent notion of which systems should, and which should not,
be distinguished, one can then justify a particular notion of equivalence, e.g. bisimulation,
by showing that it captures precisely the touchstone equivalence. In other words, a particular definition of bisimulation
is appropriate because $\approx_{\text{\tiny bisi}}$, the associated bisimulation equivalence,
\begin{enumerate}[(i)]
\item is \emph{sound} with respect to the touchstone equivalence, that is
$s_1 \bisim s_2$ implies $s_1 \approx_{\text{\tiny behav}} s_2$; 
%\Y{[Commas needed for enumerated items and single-line mathematical equations. The same comment applies throughout the  paper.]}

\item  provides a \emph{complete} proof methodology for the
touchstone equivalence, that is $s_1 \approx_{\text{\tiny behav}} s_2$ implies
$s_1 \bisim s_2$.
\end{enumerate}

This approach originated in \cite{ht92} but has now been
widely used for different process description languages; for example,
see \cite{JeffreyRathke05hopi,EnvBisim07} for its application to
higher-order process languages, \cite{RathkeS08} for mobile ambients,
\cite{FournetG05} for asynchronous languages and \cite{DH11a} for probabilistic timed languages. Moreover, in each
case the distinguishing criteria are more or less the same. The
touchstone equivalence should
\begin{enumerate}[(i)]
\item be \emph{compositional}; that is preserved by some natural operators for
constructing systems;

\item \emph{preserve barbs}; barbs are simple experiments which observers may perform
on systems \cite{julian};

\item be \emph{reduction-closed}; this is a natural condition on the reduction semantics of
systems which ensures that  nondeterministic choices are in some sense preserved.
\end{enumerate}

We adapt this approach to quantum processes. Using natural versions of these criteria
we obtain an appropriate touchstone equivalence, which we call
\emph{reduction barbed congruence}, $\rbc$. We then develop a theory
of bisimulations which is both sound and complete for $\rbc$.
Moreover, we provide a modal characterisation of $\rbc$ in
a quantum logic based on Hennessy-Milner logic \cite{HM85}by establishing the coincidence of the largest bisimilation with logical equivalence.

The remainder of the paper is organised as follows.  In the next section we recall some preliminary concepts from quantum theory. In Section~\ref{sec:plts} we review the model of probabilistic labelled transition systems, based on which we give the operational semantics of qCCS in Section~\ref{sec:qccs}. Section~\ref{sec:scs} contains the main theoretical results of the paper. We define a notion of open bisimulation, which is shown to be a congruence relation in the language of qCCS. It turns out that open bisimilarity precisely captures reduction barbed congruence, thus provides a sound and complete proof methodology for our touchstone equivalence. In addition, we give a modal characterisation of the equivalence in a quantum logic obtained by an extension of Hennessy-Milner logic with a probabilistic choice modality and a super-operator application modality. To illustrate the application of open bisimulation and its modal characterisation, in Section~\ref{sec:examples} we describe the key distribution protocol BB84 as qCCS processes and compare a specification with its implementations of the protocol. The paper ends with a brief comparison with related work in Section~\ref{sec:conl}.
%%%%%%%%%%%%%
\section{Preliminaries on quantum mechanics}
In this section, we briefly recall some basic concepts from quantum theory, which requires first  some notions from linear algebra. 
More details about quantum computation can be found in many books, e.g. \cite{NC00}.

\subsection{Basic linear algebra}
A {\it Hilbert space} $\h$ is a complete vector space equipped with an inner
product $$\langle\cdot|\cdot\rangle:\h\times \h\rightarrow \mathbf{C}$$
such that 
\begin{enumerate}
\item
$\langle\psi|\psi\rangle\geq 0$ for any $|\psi\>\in\h$, with
equality if and only if $|\psi\rangle =0$;
\item
$\langle\phi|\psi\rangle=\langle\psi|\phi\rangle^{\ast}$;
\item
$\langle\phi|\sum_i c_i|\psi_i\rangle=
\sum_i c_i\langle\phi|\psi_i\rangle$,
\end{enumerate}
where $\mathbf{C}$ is the set of complex numbers, and for each
$c\in \mathbf{C}$, $c^{\ast}$ stands for the complex
conjugate of $c$. For any vector $|\psi\rangle\in\h$, its
length $|||\psi\rangle||$ is defined to be
$\sqrt{\langle\psi|\psi\rangle}$, and it is said to be {\it normalized} if
$|||\psi\rangle||=1$. Two vectors $|\psi\>$ and $|\phi\>$ are
{\it orthogonal} if $\<\psi|\phi\>=0$. An {\it orthonormal basis} of a Hilbert
space $\h$ is a basis $\{|i\rangle\}$ where each $|i\>$ is
normalized and any pair of them are orthogonal.

Let $\lh$ be the set of linear operators on $\h$.  For any $A\in
\lh$, $A$ is {\it Hermitian} if $A^\dag=A$ where
$A^\dag$ is the adjoint operator of $A$ such that
$\<\psi|A^\dag|\phi\>=\<\phi|A|\psi\>^*$ for any
$|\psi\>,|\phi\>\in\h$. The fundamental {\it spectral theorem} \cite{NC00} states that
the set of all normalized eigenvectors of a Hermitian operator in
$\lh$ constitutes an orthonormal basis for $\h$. That is, there exists
a so-called spectral decomposition for each Hermitian $A$ such that
$$A=\sum_i\lambda_i |i\>\<i|=\sum_{\lambda_{i}\in spec(A)}\lambda_i E_i$$
where the set $\{|i\>\}$ constitutes an orthonormal basis of $\h$, $spec(A)$ denotes the set of
eigenvalues of $A$,
and $E_i$ is the projector to
the corresponding eigenspace of $\lambda_i$.
A linear operator $A\in \lh$ is {\it unitary} if $A^\dag A=A A^\dag=I_\h$, where $I_\h$ is the
identity operator on $\h$. 
%In  this paper, we will use some well-known unitary operators listed
%as follows: the quantum control-not operator performed on two qubits with the matrix representation 
%$$CN=\left(%
%\begin{array}{cccc}
% 1 & 0 & 0 & 0 \\
%  0 & 1 & 0 & 0 \\
%  0 & 0 & 0 & 1 \\
%  0 & 0 & 1 & 0
%\end{array}%
%\right)$$ under the computational basis, and 
For instance, a well-known unitary operator is the 1-qubit Hadamard
operator $H$ defined as follows:
% and Pauli operators $\sigma^0,\sigma^1,\sigma^2,\sigma^3$ defined respectively as
\[
H=\frac{1}{\sqrt{2}}\left(%
\begin{array}{cc}
  1 & 1 \\
  1 & -1 \\
\end{array}%
\right). %,\ \  \sigma^0=I=\left(%
%\begin{array}{cc}
%  1 & 0 \\
%  0 & 1 \\
%\end{array}%
%\right),
\]

%\[
%\sigma^1=\left(%
%\begin{array}{cc}
%  0 & 1 \\
%  1 & 0 \\
%\end{array}%
%\right),\ \sigma^2=\left(%
%\begin{array}{cc}
%  1 & 0 \\
%  0 & -1 \\
%\end{array}%
%\right),\ \sigma^3=\left(%
%\begin{array}{cc}
%  0 & -i \\
%  i & 0 \\
%\end{array}%
%\right).
%\]

The {\it  trace} of $A\in\lh$ is defined as $\tr(A)=\sum_i \<i|A|i\>$ for some
given orthonormal basis $\{|i\>\}$ of $\h$. It is worth noting that
trace function is actually independent of the chosen orthonormal basis. It is also easy to check that trace function is linear and
$\tr(AB)=\tr(BA)$ for any operators $A,B\in \lh$.

Let $\h_1$ and $\h_2$ be two Hilbert spaces. Their {\it tensor product} $\h_1\otimes \h_2$ is
defined as a vector space consisting of
linear combinations of the vectors
$|\psi_1\psi_2\rangle=|\psi_1\>|\psi_2\rangle =|\psi_1\>\otimes
|\psi_2\>$ with $|\psi_1\rangle\in \h_1$ and $|\psi_2\rangle\in
\h_2$. Here the tensor product of two vectors is defined by a new
vector such that
$$\left(\sum_i \lambda_i |\psi_i\>\right)\otimes
\left(\sum_j\mu_j|\phi_j\>\right)=\sum_{i,j} \lambda_i\mu_j
|\psi_i\>\otimes |\phi_j\>.$$ Then $\h_1\otimes \h_2$ is also a
Hilbert space where the inner product is defined as the following:
for any $|\psi_1\>,|\phi_1\>\in\h_1$ and $|\psi_2\>,|\phi_2\>\in
\h_2$,
$$\<\psi_1\otimes \psi_2|\phi_1\otimes\phi_2\>=\<\psi_1|\phi_1\>_{\h_1}\<
\psi_2|\phi_2\>_{\h_2}$$ where $\<\cdot|\cdot\>_{\h_i}$ is the inner
product of $\h_i$. For any $A_1\in \mathcal{L}(\h_1)$ and $A_2\in
\mathcal{L}(\h_2)$, $A_1\otimes A_2$ is defined as a linear operator
in $\mathcal{L}(\h_1 \otimes \h_2)$ such that for each
$|\psi_1\rangle \in \h_1$ and $|\psi_2\rangle \in \h_2$,
$$(A_1\otimes A_2)|\psi_1\psi_2\rangle = A_1|\psi_1\rangle\otimes
A_2|\psi_2\rangle.$$  The {\it partial trace} of $A\in\mathcal{L}(\h_1
\otimes \h_2)$ with respected to $\h_1$ is defined as
$\tr_{\h_1}(A)=\sum_i \<i|A|i\>$ where $\{|i\>\}$ is an orthonormal
basis of $\h_1$. Similarly, we can define the partial trace of $A$
with respected to $\h_2$. Partial trace functions are also
independent of the orthonormal basis selected.

An operator $A\in\lh$ is \emph{positive} if $\<\psi|A\psi\> \geq 0$ for every $\psi\in\h$.
A linear operator $\e$ on $\lh$ is {\it completely positive} if it maps
positive operators in $\mathcal{L}(\h)$ to positive operators in
$\mathcal{L}(\h)$, and for any auxiliary Hilbert space $\h'$, the
trivially extended operator $\mathcal{I}_{\h'}\otimes \e$ also maps
positive operators in $\mathcal{L(H'\otimes H)}$ to positive
operators in $\mathcal{L(H'\otimes H)}$. Here $\mathcal{I}_{\h'}$ is
the identity operator on $\mathcal{L(H')}$. The elegant and powerful
{\it Kraus representation theorem} \cite{Kr83} of completely positive
operators states that a linear operator $\e$ is completely positive
if and only if there is some set of operators $\{E_i\}$ with appropriate dimension such that
$$
\e(A)=\sum_{i} E_iA E_i^\dag
$$
for any $A\in \lh$. The operators $E_i$ are called Kraus operators
of $\e$. A linear operator is said to be a {\it super-operator} if it is
completely positive and trace-nonincreasing. Here an operator $\e$ is
{\it trace-nonincreasing} if $\tr(\e(A))\leq \tr(A)$ for any positive $A\in \lh$, and it is said to be
{\it trace-preserving} if the equality always holds. Then a super-operator (resp. a 
trace-preserving super-operator) is 
a completely positive operator
with its Kraus operators $E_i$ satisfying $\sum_i E_i^\dag E_i\leq I$ (resp. $\sum_i E_i^\dag E_i= I$).
We denote by $\so(\CH)$ the set of trace-preserving super-operators on the Hilbert space $\CH$.

\subsection{Basic quantum mechanics}

According to von Neumann's formalism of quantum mechanics
\cite{vN55}, an isolated physical system is associated with a
Hilbert space which is called the {\it state space} of the system. A {\it pure state} of a
quantum system is a normalized vector in its state space, and a
{\it mixed state} is represented by a density operator on the state
space. Here a density operator $\rho$ on Hilbert space $\h$ is a
positive linear operator such that $\tr(\rho)= 1$. 
Another
equivalent representation of density operator is probabilistic
ensemble of pure states. In particular, given an ensemble
$\{(p_i,|\psi_i\rangle)\}$ where $p_i \geq 0$, $\sum_{i}p_i=1$,
and $|\psi_i\rangle$ are pure states, then
$\rho=\sum_{i}p_i[|\psi_i\rangle]$ is a density
operator. Here $[|\psi_i\rangle]$ denotes the abbreviation of
$|\psi_i\>\langle\psi_i|$. Conversely, each density operator can be generated by an
ensemble of pure states in this way.  The set of
density operators on $\h$ can be defined as
$$\dh=\{\ \rho\in\lh\ :\  \rho\mbox{ is positive and } \tr(\rho)=
\mbox{1}\}.$$ 

The state space of a composite system (for example, a quantum system
consisting of many qubits) is the tensor product of the state spaces
of its components. For a mixed state $\rho$ on $\h_1 \otimes \h_2$,
partial traces of $\rho$ have explicit physical meanings: the
density operators $\tr_{\h_1}\rho$ and $\tr_{\h_2}\rho$ are exactly
the reduced quantum states of $\rho$ on the second and the first
component system, respectively. Note that in general, the state of a
composite system cannot be decomposed into tensor product of the
reduced states on its component systems. A well-known example is the
 2-qubit state
$$|\Psi\>=\frac{1}{\sqrt{2}}(|00\>+|11\>).
$$
%which appears repeatedly in our examples of this paper. 
This kind of state is called {\it entangled state}.
To see the strangeness of entanglement, suppose a measurement $M=
\lambda_0[|0\>]+\lambda_1[|1\>]$ is applied on the first qubit
of $|\Psi\>$ (see the following for the definition of
quantum measurements). Then after the measurement, the second qubit will
definitely collapse into state $|0\>$ or $|1\>$ depending on whether
the outcome $\lambda_0$ or $\lambda_1$ is observed. In other words,
the measurement on the first qubit changes the state of the second
qubit in some way. This is an outstanding feature of quantum mechanics
which has no counterpart in classical world, and is the key to many
quantum information processing tasks  such as teleportation
\cite{BB93} and superdense coding \cite{BW92}.

The evolution of a closed quantum system is described by a unitary
operator on its state space: if the states of the system at times
$t_1$ and $t_2$ are $\rho_1$ and $\rho_2$, respectively, then
$\rho_2=U\rho_1U^{\dag}$ for some unitary operator $U$ which
depends only on $t_1$ and $t_2$. In contrast, the general dynamics which can occur in a physical system is
described by a trace-preserving super-operator on its state space. 
Note that the unitary transformation $U(\rho)=U\rho U^\dag$ is
a trace-preserving super-operator. 

A quantum {\it measurement} is described by a
collection $\{M_m\}$ of measurement operators, where the indices
$m$ refer to the measurement outcomes. It is required that the
measurement operators satisfy the completeness equation
$\sum_{m}M_m^{\dag}M_m=I_\h$. If the system is in state $\rho$, then the probability
that measurement result $m$ occurs is given by
$$p(m)=\tr(M_m^{\dag}M_m\rho),$$ and the state of the post-measurement system
is $M_m\rho M_m^{\dag}/p(m).$ 

A particular case of measurement is {\it projective measurement} which is usually represented by a Hermitian operator.  Let  $M$ be a
Hermitian operator and
\begin{equation}\label{eq:specdec}
M=\sum_{m\in spec(M)}mE_m
\end{equation} 
its spectral decomposition. Obviously, the projectors  $\{E_m:m\in
spec(M)\}$ form a quantum measurement. If the state of a quantum
system is $\rho$, then the probability that result $m$ occurs when
measuring $M$ on the system is $p(m)=\tr(E_m\rho),$ and the
post-measurement state of the system is $E_m\rho E_m/p(m).$
Note that for each outcome $m$, the map $$\e_m(\rho) =
E_m\rho E_m$$
is again a super-operator by Kraus Theorem; it is not
trace-preserving in general.

Let $M$ be a projective measurement with Eq.(\ref{eq:specdec}) its spectral decomposition. We call $M$ non-degenerate if for any $m\in spec(M)$, the corresponding projector $E_{m}$ is 1-dimensional; that is, all eigenvalues of $M$ are non-degenerate. Non-degenerate measurement is obviously a very special case of general quantum measurement. However, when an ancilla system lying at a fixed state is provided, non-degenerate measurements together with unitary operators are sufficient to implement general measurements~\cite{NC00}.

%%%%%%%%%%%%%
\section{A probabilistic model}
\label{sec:plts}

In this section we review the model of probabilistic labelled
transition systems (pLTSs), and some properties of weak transitions. Later on we will interpret the behaviour
of quantum processes in terms of pLTSs.

\subsection{Probabilistic labelled transition systems}
\label{subsec:plts}

We begin with some notation. A (discrete) probability distribution
over a set S is a function $\Delta : S \rightarrow [0, 1] $ with
$\sum_{s\in S} \Delta(s) = 1$; the support of such a $\Delta$ is
the set $\support{\Delta} = \setof{s \in S}{\Delta(s) > 0}$.  
The point distribution $\pdist{s}$ assigns probability
$1$ to $s$ and $0$ to all other elements of $S$, so that
$\support{\pdist{s}} = s$. We use $\dist{S}$ to denote the set of
distributions over $S$, ranged over by $\Delta,\Theta$ etc.
If $\sum_{k \in K} p_k = 1$ for some
collection of  $p_k \geq 0$, and the $\Delta_k$ are distributions,
then so is $\sum_{k \in K}p_k \cdot \Delta_k$ with
$(\sum_{k \in K}p_k \cdot \Delta_k)(s)~=~\sum_{i\in I} p_i\cdot \Delta_i(s).$

\begin{definition}\label{def:LTS}
  A \emph{probabilistic labelled transition system} (pLTS) is a triple
$\langle S, \Act_\tau,  \rightarrow  \rangle$, where
\begin{enumerate}[(i)] \parskip 0pt
\item $S$ is a set of states; 
\item $\Act_\tau$ is a set of transition labels, with distinguished element $\tau$;
\item the relation $\rightarrow$ is a subset of
$S \times \Act_{\tau}  \times \dist{S}$.
\end{enumerate}
\end{definition}

In the literature essentially the same model has appeared under different names such as \emph{NP-systems}
\cite{JHW94}, \emph{probabilistic processes} \cite{JW95},
\emph{simple probabilistic automata} \cite{Seg95},
\emph{probabilistic transition systems} \cite{JW02} etc.
Furthermore, there are strong structural similarities with \emph{Markov
  Decision Processes} \cite{Put94,scalar}.

A (non-probabilistic) labelled transition system (LTS) may be viewed
as a degenerate pLTS, one in which only point distributions are
used.

%A pLTS is \emph{finitary} if the state set $S$ is finite and for each $s\in S$ the set $\sset{(\alpha,\Delta)\mid s\ar{\alpha}\Delta}$ is finite; in this paper we are primarily concerned with finitary pLTSs.

\subsection{Lifting relations}

In a pLTS actions are only performed by states, in that actions are
given by relations from states to distributions. But in general we
allow distributions over states to perform an action. For this
purpose, we \emph{lift} these relations so that they also apply to
distributions \cite{DGHM09}.
\begin{definition}[Lifting]\label{def:lift}
Let  $\mathord{\aRel} \subseteq
  S\times\dist{S}$ be a relation from states to distributions in a pLTS.
Then $\mathord{\lift{\aRel}} \subseteq \dist{S} \times
\dist{S}$ is the smallest relation that satisfies
\begin{enumerate}[(i)]\itemsep 0pt

\item $s \aRel \Theta$ implies $\pdist{s} \lift{\aRel} \Theta$, and
\item\label{i1529} (Linearity)
 $\Delta_i \lift{\aRel} \Theta_i$ for $i\in I$ implies
 $(\sum_{i\in I}p_i\cdot\Delta_i)\lift{\aRel}(\sum_{i\in I}p_i\cdot\Theta_i)$
 for any $p_i \in [0,1]$ with $\sum_{i\in I}p_i = 1$, where $I$ is a
 finite index set.
%\qed
\end{enumerate}
\end{definition}
\noindent
%Note that the definition of linearity uses only a finite index set $I$; this is sufficient for our purposes as our primary focus is on finite state systems. Indeed in the remainder of the paper all index sets can be taken to be finite, unless indicated otherwise.

There are numerous ways of formulating this concept of lifting relations. The following is
particularly useful.
\begin{lemma}\label{lem:lifting.alt}
  $\Delta \lift{\aRel} \Theta$ if and only if there is a finite index set $I$ such that
    \begin{enumerate}[(i)]
    \item $\Delta = \sum_{i \in I} p_i \cdot \pdist{s_i}$,
    \item $\Theta = \sum_{i \in I} p_i \cdot \Theta_i$,
    \item $s_i \aRel \Theta_i$ for each $i \in I$.
    \end{enumerate}
\end{lemma}
\begin{proof}
  ($\Leftarrow$) Suppose there is an index set $I$ such that (i)
  $\Delta = \sum_{i \in I} p_i \cdot \pdist{s_i}$, (ii) $\Theta = \sum_{i \in
    I} p_i \cdot \Theta_i$, and (iii) $s_i \aRel \Theta_i$ for each $i
  \in I$. By (iii) and the first rule in Definition~\ref{def:lift}, we
  have $\pdist{s_i} \lift{\aRel}\Theta_i$ for each $i\in I$. By the
  second rule in  Definition~\ref{def:lift} we  obtain that
$(\sum_{i\in I} p_i\cdot \pdist{s_i}) \lift{\aRel} (\sum_{i\in
  I}p_i\cdot\Theta_i)$, that is $\Delta\lift{\aRel}\Theta$.

($\Rightarrow$)  We proceed by rule induction.
\begin{itemize}
\item If $\Delta \lift{\aRel} \Theta$ because of $\Delta=\pdist{s}$
  and $s\aRel\Theta$, then we can simply take $I$ to be the singleton
  set $\sset{i}$ with
  $p_i=1$ and $\Theta_i=\Theta$.

\item If $\Delta \lift{\aRel} \Theta$ because of the conditions
  $\Delta=\sum_{i\in I}p_i\cdot\Delta_i$, $\Theta_i=\sum_{i\in
    I}p_i\cdot\Theta_i$ for some index set $I$, and $\Delta_i
  \lift{\aRel} \Theta_i$ for each $i\in I$, then by induction
  hypothesis there are index sets $J_i$ such that $\Delta_i=\sum_{j\in J_i}p_{ij}\cdot \pdist{s_{ij}}$,\
  $\Theta_i=\sum_{j\in J_i}p_{ij}\cdot \Theta_{ij}$, and
  $s_{ij}\aRel\Theta_{ij}$ for each $i\in I$ and $j\in J_i$. It follows that 
$\Delta=\sum_{i\in I}\sum_{j\in J_i}p_ip_{ij}\cdot \pdist{s_{ij}}$,
$\Theta=\sum_{i\in I}\sum_{j\in J_i}p_ip_{ij}\cdot \Theta_{ij}$, and
$s_{ij}\aRel\Theta_{ij}$ for each $i\in I$ and $j\in J_i$. So it
suffices to take $\sset{ij \mid i\in I,j\in J_i}$ to be the index set
and $\sset{p_ip_{ij} \mid i\in I, j\in J_i}$ be the collection of probabilities.
%\hfill\qed
\end{itemize}
\end{proof}

We apply this operation to the relations  $\ar{\alpha}$ in the pLTS
for $\alpha\in \Act_{\tau}$, where we also write $\ar{\alpha}$ for
$\lift{(\ar{\alpha})}$. Thus as source of
a relation $\ar{\alpha}$ we now also allow distributions.
But note that  $\pdist{s} \ar{\alpha} \Delta$ is more general than 
$s \ar{\alpha} \Delta$. In papers such as \cite{segalaL95,DP07} the
former is refered to as a \emph{combined transition} because if
$\pdist{s}\ar{\alpha}\Delta$ then there is a collection of distributions
$\Delta_i$ and probabilities $p_i$ such that $s\ar{\alpha}\Delta_i$ for each $i\in I$ and
$\Delta=\sum_{i\in I}p_i\cdot\Delta_i$ with $\sum_{i\in I}p_i=1$.

In Definition~\ref{def:lift}, linearity tells us how to compare two
linear combinations of distributions. Sometimes we need a dual
notion of decomposition. Intuitively, if a relation $\aRel$ is
\emph{left-decomposable} and $\Delta \aRel \Theta$, then for any
decomposition of $\Delta$ there exists some corresponding
decomposition of $\Theta$.
\begin{definition}[Left-decomposable]\label{def:left-deconst}
 A binary relation over distributions,  $\mathord{\aRel} \subseteq
  \dist{S} \times\dist{S}$,
 is called \emph{ left-decomposable } if
 $(\sum_{i\in I}p_i\cdot\Delta_i)   \aRel \Theta$,
 where $I$ is a finite index set,
implies that
   $\Theta$ can be written as $(\sum_{i\in I}p_i\cdot\Theta_i)$
    such that $\Delta_i \aRel \Theta_i$ for every $i \in I$.
\end{definition}

\begin{proposition}\label{prop:act.dstruct}
  For any $\aRel \;\subseteq\;S \times \dist{S}$ the relation 
$\lift{\aRel}$ over distributions is left-decomposable. 
\end{proposition}
\begin{proof}
Suppose
$ \Delta = (\sum_{i\in I}p_i\cdot\Delta_i)$ and $\Delta    \;\lift{\aRel}\; \Theta$.
We have to find a family of $\Theta_i$ such that
\begin{enumerate}[(i)]
\item $\Delta_i \lift{\aRel} \Theta_i$ for each $i\in I$,

\item $\Theta  = \sum_{i\in I}p_i\cdot\Theta_i$.
\end{enumerate}
From the alternative characterisation of lifting,
Lemma~\ref{lem:lifting.alt}, we know that
\[
  \Delta = \sum_{j \in J} q_j \cdot \pdist{s_j}  \qquad s_j \aRel \Theta^j  \qquad  \Theta = \sum_{j \in J} q_j \cdot \Theta^j
\]
Define $\Theta_i$ to be 
$$
\sum_{s \in \support{\Delta_i}} \Delta_i(s) \cdot (\sum_{\setof{j \in J}{s=s_j}} \frac{q_j}{\Delta(s)} \cdot \Theta^j)
$$
Note that $\Delta(s)$ can be written as $\sum_{\setof{j \in J}{s=s_j}} q_j$ and therefore 
$$
\Delta_i  = \sum_{s \in \support{\Delta_i}} \Delta_i(s) 
            \cdot (\sum_{\setof{j \in J}{s=s_j}} \frac{q_j}{\Delta(s)} \cdot\pdist{s_j})
$$
Since $s_j \aRel \Theta_j$ this establishes (i) above. 

To establish (ii) above let us first abbreviate the sum 
$\sum_{\setof{j \in J}{s=s_j}} \frac{q_j}{\Delta(s)} \cdot \Theta^j$ to $X(s)$. 
Then $\sum_{i\in I} p_i \cdot \Theta_i$ can be written as 
\[\begin{array}{rl}
  & \sum_{s \in\support{\Delta}} \sum_{i \in I} p_i \cdot \Delta_i(s) \cdot X(s)\\
=  & \sum_{s \in\support{\Delta}} (\sum_{i \in I} p_i \cdot \Delta_i(s) ) \cdot X(s)\\
=  &  \sum_{s \in\support{\Delta}} \Delta(s) \cdot X(s)
\end{array}\]
The last equation is justified by the fact that $\Delta(s) = \sum_{i \in I} p_i \cdot \Delta_i(s)$.

Now $\Delta(s) \cdot X(s) = \sum_{\setof{j \in J}{s=s_j}} q_j \cdot \Theta^j$ and therefore we have
\[\begin{array}{rl}
  \sum_{i\in I} p_i \cdot \Theta_i &= 
   \sum_{s \in\support{\Delta}}  \sum_{\setof{j \in J}{s=s_j}} q_j \cdot \Theta^j \\
             &=\sum_{j \in J} q_j \cdot \Theta^j\\
             &= \Theta
\end{array}\]
%\hfill\qed
\end{proof}

\noindent
%As a consequence we can now assume that the action relations $\Delta \ar{\alpha} \Theta$ over distributions are both linear and left-decomposable. 

%\subsection{Weak transitions}

We write $s \ar{\hat{\tau}} \Delta$ if either $s \ar{\tau}
\Delta$ or $\Delta = \pdist{s}$, and  $s \ar{\hat{a}} \Delta$ iff
$s\ar{a}\Delta$ for $a\in\Act$. For any $a\in\Act_\tau$, we know
that $\mathord{\ar{\hat{a}}}\subseteq S\times\dist{S}$, so we can
lift it to be a transition relation between distributions. With a
slight abuse of notation we simply write $\Delta\ar{\hat{a}}\Theta$
for $\Delta\lift{(\ar{\hat{a}})}\Theta$. Then we define weak
transitions $\dar{\hat{a}}$ by letting $\dar{\hat{\tau}}$ be the
reflexive and transitive closure of $\ar{\hat{\tau}}$ and writing
$\Delta\dar{\hat{a}} \Theta$ for $a \in\Act$ whenever $\Delta
\dar{\hat{\tau}} \ar{\hat{a}} \dar{\hat{\tau}} \Theta$.
If $\Delta$ is a point distribution, we often write 
$s\dar{\hat{a}}\Theta$ instead of $\pdist{s}\dar{\hat{a}}\Theta$.

\begin{proposition}\label{prop:linear.move}
 The action relations $\dar{\hat{\alpha}}$ are both linear and 
left-decomposable.
\end{proposition}
\begin{proof}
  It is easy to check that both properties are preserved by composition; that is if $\aRel_i, i = 1,2$, 
  are linear, left-decomposable respectively, then so is $\aRel_1 \cdot \aRel_2$. 
The result now follows since $\dar{\hat{\alpha}}$ is formed by repeated composition from two relations $\ar{\hat{\tau}}$ and $\ar{\hat{\alpha}}$ which we know
are both linear and left-decomposable. 
%\hfill\qed
\end{proof}
 
Let $\aRel\ \subseteq S \times S$ be a relation between states. It induces a speical relation $\hat{\aRel}\subseteq S\times\dist{S}$ between states and distributions:
\[\hat{\aRel} \Defs \sset{(s,\pdist{t})\mid s\aRel t}.\]
Then we can use Definition~\ref{def:lift} to lift $\hat{\aRel}$ to be a relation $\lift{(\hat{\aRel})}$ between distributions. For simplicity, we combine the above two lifting operations and directly write $\lift{\aRel}$ for $\lift{(\hat{\aRel})}$ in the sequel, with the intention that a relation between states can be lifted to a relation between distributions via a special application of Definition~\ref{def:lift}.
Consequently, we have the following corollary of
Lemma~\ref{lem:lifting.alt}.
\begin{corollary}\label{cor:lifting.states}
Suppose $\aRel\ \subseteq S\times S$. Then
  $\Delta \lift{\aRel} \Theta$ if and only if there is a finite index set $I$ such that
    \begin{enumerate}[(i)]
    \item $\Delta = \sum_{i \in I} p_i \cdot \pdist{s_i}$,
    \item $\Theta = \sum_{i \in I} p_i \cdot \pdist{t_i}$,
    \item $s_i \aRel t_i$ for each $i \in I$.
    \end{enumerate}
\end{corollary}

Relations over distributions obtained by \emph{lifting} enjoy some very useful
properties. %which we encapsulate in the next two propositions.
The following one will be used in Section~\ref{sec:scs} to show the transitivity of open bisimilarity.
\begin{proposition}\label{prop:lift2}
Let $\aRel_1,\aRel_2\subseteq S\times S$ be two binary relations. The forward relation $\lift{(\aRel_1\cdot\aRel_2)}$ coincides with $\lift{\aRel_1}\cdot\lift{\aRel_2}$.
\end{proposition}
\begin{proof}
We first show that $\lift{(\aRel_1\cdot\aRel_2)} ~\subseteq~
\lift{\aRel_1}\cdot\lift{\aRel_2}$. Suppose there are two distributions
$\Delta_1,\Delta_2$ such that
$\Delta_1\lift{(\aRel_1\cdot\aRel_2)}\Delta_2$. Then we have that
 \begin{equation}\label{eq:tran}
    \Delta_1 = \sum_{i \in I}{p_i \cdot \pdist{s_i}},
     \qquad
    s_i\ \aRel_1\cdot\aRel_2\ s'_i,
     \qquad
   \Delta_2 = \sum_{i \in I}{ p_i \cdot \pdist{s'_i}} ~.
  \end{equation}
The middle part of (\ref{eq:tran}) implies the existence of some
states $t_i$ such that $s_i \aRel_1 t_i$ and $t_i\aRel s'_i$. Let
$\Theta$ be the distribution $\sum_{i\in I}{p_i \cdot \pdist{t_i}}$.
It is clear that $\Delta_1\lift{\aRel_1}\Theta$ and
$\Theta\lift{\aRel_2}\Delta_2$. It follows that
$\Delta_1\lift{\aRel_1}\cdot\lift{\aRel_2}\Delta_2$.

Then we show the inverse inclusion $\lift{\aRel_1}\cdot\lift{\aRel_2}
~\subseteq~ \lift{(\aRel_1\cdot\aRel_2)}$.
 Given three distributions $\Delta_1,\Delta_2,\Delta_3$, we show that if $\Delta_1 \lift{\aRel_1}
  \Delta_2$ and
  $\Delta_2 \lift{\aRel_2} \Delta_3$ then $\Delta_1 \lift{(\aRel_1\cdot\aRel_2)} \Delta_3$.

  First $\Delta_1 \lift{\aRel_1} \Delta_2$ means that
  \begin{equation}\label{eq:dec1}
    \Delta_1 = \sum_{i \in I}{p_i \cdot \pdist{s_i}},
     \qquad
    s_i\ \aRel_1\ s'_i,
     \qquad
   \Delta_2 = \sum_{i \in I}{ p_i \cdot \pdist{s'_i}}.
  \end{equation}
Then from $\Delta_2 \lift{\aRel_2} \Delta_3$ and Proposition~\ref{prop:act.dstruct},
we have $\Delta_3 = \sum_{i \in I}{ p_i \cdot \Theta_i}$ with $\pdist{s'_i}\lift{\aRel_2} \Theta_i$ for each $i\in I$.
Now by Corollary~\ref{cor:lifting.states}, $\Theta_i$ can be further decomposed as $\Theta_i=\sum_{j\in J_i}q_{ij}\cdot \pdist{t_{ij}}$ such that
$s'_i\aRel_2 t_{ij}$ for each $j\in J_i$. In summary, we have
  \begin{equation}\label{eq:dec2}
    \Delta_1 = \sum_{i \in I}{p_i \cdot \sum_{j\in J_i} q_{ij}\cdot\pdist{s_i}},
     \qquad
    \mbox{ and }
     \qquad
   \Delta_3 = \sum_{i \in I}{p_i \cdot \sum_{j\in J_i} q_{ij}\cdot\pdist{t_{ij}}}.
  \end{equation}
Finally, it follows from (\ref{eq:dec2}) and the fact $s_i \aRel_1 s'_i\aRel_2 t_{ij}$ that $\Delta_1 \lift{(\aRel_1\cdot\aRel_2)} \Delta_3$.
\end{proof}

%%%%%%%%%%%%%
\section{Quantum CCS}\label{sec:qccs}

We introduce the language qCCS which was originally studied in \cite{FDJY07, YFDJ09, FDY11}. Three types of data are considered in qCCS: as classical data we have \texttt{Bool} for booleans and \texttt{Real} for real numbers, and as quantum data we have \texttt{Qbt} for qubits. Consequently,
two countably infinite sets of variables are assumed: $cVar$ for classical variables, ranged over by $x,y,...$, and $qVar$ for quantum variables, ranged over by $q,r,...$. 
We assume a set $Exp$, which includes $cVar$ as a subset and is ranged over by $e,e',\dots$,  of classical data expressions over
\texttt{Real}, and a set of boolean-valued expressions $BExp$, ranged over by $b, b',\dots$, with the usual  boolean constants $\texttt{true}$, $\texttt{false}$, and operators
$\neg$, $\wedge$, $\vee$, and $\ra$. In particular, we let $e\bowtie e'$ be a boolean expression for any $e,e'\in Exp$ and $\bowtie \in\sset{>, <, \geq, \leq, =}$.
We further assume that only classical variables can occur free in both data expressions and boolean expressions.
Two types of channels are used: $cChan$ for classical channels, ranged over by $c,d,...$, and $qChan$ for quantum channels, ranged over by \texttt{c,d},.... A relabelling function $f$ is a map on $cChan\; \cup\; qChan$ such that $f(cChan)\subseteq cChan$ and $f(qChan)\subseteq qChan$.
Sometimes we abbreviate a sequence of distinct variables $q_1,...,q_n$ into $\tilde{q}$. 

The terms in qCCS are given by:
\[\begin{array}{rcl}
  P,Q &::=& \Cnil \BNFsep \tau.P \BNFsep c?x.P \BNFsep c!e.P
\BNFsep \qc?q.P \BNFsep \qc!q.P \BNFsep \CE[\tilde{q}].P
\BNFsep M[\tilde{q};x].P \\
& & \BNFsep P + Q \BNFsep \;P \Cpar Q\;
\BNFsep P[f] \BNFsep P\backslash L \BNFsep \ifthen{b}{P}
                    \BNFsep A(\tilde{q};\tilde{x})  
\end{array}\]
where $f$ is a relabelling function and $L\subseteq cChan\cup qChan$ is a set of channels.
Most of the constructors are standard as in CCS \cite{ccs}. 
We briefly explain a few new constructors. The process $\qc?q.P$ receives a quantum datum along quantum channel $\qc$ and evolves into $P$. The process $\qc!q.P$ sends out a quantum datum along quantum channel $\qc$ before evolving into $P$. The new symbols $\CE$ and $M$ represent respectively a trace-preserving super-operator and a non-degenerate projective measurement applying on the Hilbert space associated with the systems $\tilde{q}$. 

Free classical variables can be defined in the usual way, except for the fact that the variable $x$ in the quantum measurement $M[\tilde{q};x]$ is bound. A process $P$ is closed if it contains no free classical variable, i.e. $fv(P)=\emptyset$.

The set of free quantum variables for process $P$, denoted by $qv(P)$ can be inductively defined as follows.
\[\begin{array}{rclrcl}
qv(\Cnil) & = & \emptyset & qv(\tau.P) & = & qv(P)\\
qv(c?x.P) & = & qv(P) & qv(c!e.P) & = & qv(P) \\
qv(\qc?q.P) & = & qv(P)-\sset{q} & qv(\qc!q.P) & = & qv(P)\cup\sset{q} \\
qv(\CE[\tilde{q}].P) & = & qv(P)\cup \tilde{q} &
qv(M[\tilde{q};x].P) & = & qv(P) \cup \tilde{q}\\
qv(P+Q) & = & qv(P)\cup qv(Q) \qquad &
qv(P \Cpar Q) & = & qv(P)\cup qv(Q)\\
qv(P[f]) & = & qv(P) &
qv(P\backslash L) & = & qv(P)\\
qv(\ifthen{b}{P}) & = & qv(P) &
qv(A(\tilde{q};\tilde{x})) & = & \tilde{q}.
\end{array}\]
For a process to be legal, we require that 
\begin{enumerate}
\item $q\not\in qv(P)$ in the process $\qc!q.P$;
\item $qv(P)\cap qv(Q)=\emptyset$ in the process $P \Cpar Q$;
\item Each constant $A(\tilde{q};\tilde{x})$ has a defining equation $A(\tilde{q};\tilde{x}) \Defs P$, where $P$ is a term with $qv(P)\subseteq\tilde{q}$ and $fv(P)\subseteq \tilde{x}$.
\end{enumerate}
The first condition says that a quantum system will not be referenced after it has been sent out. This is a requirement of quantum no-cloning theorem. The second condition says that parallel composition $\Cpar$ models separate parties that never reference a quantum system simultaneously.

Throughout the paper we implicitly assume the convention that processes are identified up to $\alpha$-conversion, bound variables differ from each other and they are different from free variables.

We now turn to the operational semantics of qCCS. For each quantum variable $q$ we assume a 2-dimensional Hilbert space $\CH_q$. For any nonempty subset $S\subseteq qVar$ we write $\CH_S$ for the tensor product space $\bigotimes_{q\in S}\CH_q$. In particular, $\CH=\CH_{\it qVar}$ is the state space of the whole environment consisting of all the quantum variables, which is a countably infinite dimensional Hilbert space.

Let $P$ be a closed quantum process and $\rho$ a density operator on $\CH$, the pair $\pair{P,\rho}$ is called a \emph{configuration}. We write $\Con$ for the set of all configurations, ranged over by $\CC$ and $\CD$.
We interpret qCCS as a pLTS whose states are all the configurations definable in the language,
and whose arrows are determined by the rules in Figure~\ref{fig:opsem}; we have omitted the obvious
symmetric counterparts to the rules \Rlts{C-Com}, \Rlts{Q-Com}, \Rlts{Int} and
\Rlts{Sum}. 
The set of actions $\Act$ takes the form
\[\sset{c?v, c!v \mid c\in cChan, v\in\texttt{Real}}\cup \sset{\qc?r,\qc!r \mid \qc\in qChan,r\in qVar}\]
The symbol $\tau$ denotes invisible actions. We write $\Act_\tau$ for $\Act\cup\sset{\tau}$, which is ranged over by $\alpha$. We use $cn(\alpha)$ for the set of channel names in action $\alpha$. So, for example, $cn(\qc?x)=\sset{\qc}$ and $cn(\tau)=\emptyset$.

In the first eight rules in Figure~\ref{fig:opsem}, the targets of arrows are point distributions, and we use the slightly abbreviated form $\CC\ar{\alpha}\CC'$ to mean $\CC\ar{\alpha}\pdist{\CC'}$.

The rules use the obvious extension of the function $\Cpar$ on terms to configurations and distributions. To be precise,
$\CC\Cpar P$ is the configuration $\<Q\Cpar P,\rho \>$ where $\CC=\<Q,\rho\>$, and
 $\Delta\Cpar P$ is the distribution defined by:
\[(\Delta\Cpar P)(\pair{Q,\rho}) = \left\{\begin{array}{ll}
\Delta(\pair{Q',\rho}) & \mbox{if $Q=Q'\Cpar P$}\\
0 & \mbox{otherwise.}
\end{array}\right.\]
Similar extension applies to $\Delta[f]$ and $\Delta\backslash L$.

%We say a process $P$ from qCCS is \emph{finitary} if the sub-pLTS consisting of all states reachable from $\pair{P,\rho}$, for any $\rho$, is finitary, and we use \emph{finitary} qCCS to refer to the pLTS consisting of all such finitary $P$. 

\begin{figure}[t]

\[\begin{array}{ll}
  \slinfer[\Rlts{Tau}]{\pair{\tau.P, \rho} \ar{\tau} \pair{P, \rho}}
  &\linfer[\Rlts{C-Inp}]{v\in\texttt{Real}}{\pair{c?x.P,\rho} \ar{c?v} \pair{P[v/x],\rho}}
  \\[3pt]
  \linfer[\Rlts{C-Outp}]{v=\Op{e}}{\pair{c!e.P,\rho} \ar{c!v} \pair{P,\rho} }
                        
 & 
\linfer[\Rlts{C-Com}]{\pair{P_1,\rho}\ar{c?v}\pair{P'_1,\rho}\qquad
\pair{P_2,\rho}\ar{c!v}\pair{P'_2,\rho}
}
                          {\pair{P_1\Cpar P_2,\rho} \ar{\tau} \pair{P'_1\Cpar P'_2,\rho} }
 \\[3pt]
  \linfer[\Rlts{Q-inp}]{r\not\in qv(\qc?q.P)}
{\pair{\qc?q.P,\rho} \ar{\qc?r} \pair{P[r/q],\rho}}
&
\slinfer[\Rlts{Q-Outp}]{\pair{\qc!q.P,\rho} \ar{\qc!q} \pair{P,\rho} }
\\[3pt]
\linfer[\Rlts{Q-Com}]{\pair{P_1,\rho} \ar{\qc?r} \pair{P'_1,\rho} \qquad \pair{P_2,\rho}\ar{\qc!r} \pair{P'_2,\rho}}
                      {\pair{P_1\Cpar P_2, \rho } \ar{\tau} \pair{P'_1\Cpar P'_2, \rho}}
&
\slinfer[\Rlts{Oper}]{\pair{\CE[\tilde{q}].P,\rho} \ar{\tau} \pair{P,\CE_{\tilde{q}}(\rho)}}
\\[3pt]
\linfer[\Rlts{Meas}]{M=\sum_{i\in I}\lambda_i E^i \qquad p_i=tr(E^i_{\tilde{q}}\rho)}
{\pair{M[\tilde{q};x].P,\rho} \ar{\tau} \sum_{i\in I}p_i \pair{P[\lambda_i/x], E^i_{\tilde{q}}\rho E^i_{\tilde{q}} / p_i}}
&
\\
\linfer[\Rlts{Int}]{\pair{P_1,\rho} \ar{\alpha} \Delta\qquad qbv(\alpha)\cap qv(P_2)=\emptyset}{\pair{P_1 \Cpar P_2,\rho} \ar{\alpha} \Delta\Cpar P_2}
%\\[3pt]
%\linfer[\Rlts{Inp-Int}]{\pair{P_1,\rho} \ar{\qc?r} \pair{P'_1,\rho}\qquad r\not\in qv(P_2)}{\pair{P_1\Cpar P_2,\rho} \ar{\qc?r} \pair{P'_1\Cpar P_2,\rho}}
& 
\linfer[\Rlts{Sum}]{\pair{P_1,\rho} \ar{\alpha} \Delta}{\pair{P_1+P_2,\rho} \ar{\alpha} \Delta}
\\[3pt]
\linfer[\Rlts{Rel}]{\pair{P,\rho} \ar{\alpha} \Delta}{\pair{P[f],\rho} \ar{f(\alpha)} \Delta[f]}
&
\linfer[\Rlts{Res}]{\pair{P,\rho} \ar{\alpha} \Delta \qquad cn(\alpha)\cap L=\emptyset}{\pair{P\backslash L,\rho} \ar{\alpha} \Delta\backslash L}
\\[3pt]
\linfer[\Rlts{Cho}]{\pair{P,\rho} \ar{\alpha} \Delta \qquad \Op{b}=\texttt{true}}{\pair{\ifthen{b}{P},\rho} \ar{\alpha} \Delta}
&
\linfer[\Rlts{Cons}]{\pair{P[\widetilde{v}/\widetilde{x},\tilde{r}/\tilde{q}],\rho} \ar{\alpha} \Delta \qquad A(\widetilde{x}, \tilde{q})\Defs P}{\pair{A(\widetilde{v},\tilde{r}),\rho} \ar{\alpha} \Delta}\end{array}\]
 \caption{Operational semantics of qCCS\label{fig:opsem}
%\Y{[As we use $\pdist{s}$ to denote a simple distribution, should we add a bar to the right hand side of the rules such as Tau? It also applies to the notation for super-operator application.]}
 }
 \rule{\linewidth}{0.5mm}
\end{figure}

%%%%%%%%%%%%%
\section{Open bisimulations}
\label{sec:scs}
Let $\CC=\pair{P,\rho}$. We use the notation $\qv(\CC)\Defs qv(P)$ for free quantum variables and $\ptr(\CC)\Defs\tr_{qv(P)}(\rho)$ for partial traces.
%\Y{[I suggest to use the less misleading notation $\tr_P(\CC)$, instead of $\tr(\CC)$, to denote $\Defs\tr_{qv(P)}(\rho)$.]}
Let $\Delta=\sum_{i\in I}p_i\cdot\pdist{\pair{P_i,\rho_i}}$. We write $\CE(\Delta)$ for the distribution $\sum_{i\in I}p_i\cdot\pdist{\pair{P_i,\CE(\rho_i)}}$.

%\begin{lemma}\label{lem:ce.sub}
%Suppose $S_1\subseteq S_2$. If $\CE$ is a trace-preserving super-operator acting on $\CH_{S_2}$, it is also a trace-preserving super-operator acting on $\CH_{S_1}$.
%\end{lemma}

\begin{definition}\label{def:bisims}
A relation $\aRel\ \subseteq \Con\times\Con$ is a \emph{strong open simulation} if 
$\CC\aRel \CD$ implies that $\qv(\CC)=\qv(\CD)$, $\ptr(\CC)=\ptr(\CD)$, 
 and for any $\CE\in \so(\CH_{\overline{qv(\CC)}})$
\begin{itemize}
\item whenever $\CE(\CC)\ar{\alpha} \Delta$, there is some distribution $\Theta$ with $\CE(\CD)\ar{\alpha}\Theta$ and $\Delta \lift{\aRel} \Theta$. 
\end{itemize}
A relation $\aRel$ is a \emph{strong open bisimulation} if both $\aRel$ and
$\aRel^{-1}$ are strong open simulations. 
%We let $\sobisi$ be the largest strong open bisimulation.

%Two quantum processes $P$ and $Q$ are bisimilar, denoted by $P\sobisi Q$, if for any quantum state $\rho\in\CD(\CH)$ and any indexed set $\tilde{v}$ of classical values, we have \[\pair{P\sset{\tilde{v}/\tilde{x}},\rho} \obisi \pair{Q\sset{\tilde{v}/\tilde{x}},\rho}.\] Here $\tilde{x}$ is the set of free classical variables contained in $P$ and $Q$.
\end{definition}
The above definition is inspired by the work of Sangiorgi \cite{San96}, where a notion of bisimulation is defined for the $\pi$-calculus by treating name instantiation in an ``open'' style. Here we deal with super-operator application in an ``open'' style, but the instantiation of variables is in an ``early'' style because the operational semantics given in Figure~\ref{fig:opsem} is essentially an early semantics. For more variants of semantics, see e.g. \cite{pibook}.

In this paper we are mainly interested in a notion of weak open bisimulation which is like strong open bisimulation but internal transitions are abstracted away.

%\begin{lemma}\label{lem:ce.sub}
%Suppose $S_1\subseteq S_2$. If $\CE$ is a trace-preserving super-operator acting on $\CH_{S_2}$, it is also a trace-preserving super-operator acting on $\CH_{S_1}$.
%\end{lemma}

\begin{definition}\label{def:bisim}
A relation $\aRel\ \subseteq \Con\times\Con$ is a \emph{weak open simulation} if 
$\CC\aRel \CD$ implies that $\qv(\CC)=\qv(\CD)$, $\ptr(\CC)=\ptr(\CD)$, 
 and for any $\CE\in \so(\CH_{\overline{qv(\CC)}})$
\begin{itemize}
\item whenever $\CE(\CC)\ar{\alpha} \Delta$, there is some distribution $\Theta$ with $\CE(\CD)\dar{\hat{\alpha}}\Theta$ and $\Delta \lift{\aRel} \Theta$. 
\end{itemize}
A relation $\aRel$ is a \emph{weak open bisimulation} if both $\aRel$ and
$\aRel^{-1}$ are weak open simulations. We let $\obisi$ be the largest weak open bisimulation. In the sequel we will simply use open bisimulation to refer to weak open bisimulation.

Two quantum processes $P$ and $Q$ are bisimilar, denoted by $P\obisi
Q$, if for any quantum state $\rho\in\CD(\CH)$ and any indexed set
$\tilde{v}$ of classical values, we have
\[\pair{P\sset{\tilde{v}/\tilde{x}},\rho} \obisi
\pair{Q\sset{\tilde{v}/\tilde{x}},\rho}.\] Here $\tilde{x}$ is the set
of free classical variables contained in $P$ and $Q$.
\end{definition}

\subsection{A useful proof technique}
In Definition~\ref{def:bisim} super-operator application and
transitions are considered at the same time. In fact, we can separate
the two issues and approach the concept of open bisimulation in an
incremental way, which turns out to be very useful when proving that
two configurations are open bisimilar.
\begin{definition}
A relation $\aRel\subseteq \Con\times\Con$ is closed under super-operator application if $\CC\aRel\CD$ implies $\CE(\CC)\aRel \CE(\CD)$ for any $\CE\in \so(\CH_{\overline{qv(\CC)}})$.
\end{definition}

\begin{definition}
A relation $\aRel\subseteq \Con\times\Con$ is a \emph{ground simulation} if 
$\CC\aRel \CD$ implies that $\qv(\CC)=\qv(\CD)$, $\ptr(\CC)=\ptr(\CD)$, 
 and 
\begin{itemize}
\item whenever $\CC\ar{\alpha} \Delta$, there is some distribution $\Theta$ with $\CD\dar{\hat{\alpha}}\Theta$ and $\Delta \lift{\aRel} \Theta$. 
\end{itemize}
A relation $\aRel$ is a \emph{ground bisimulation} if both $\aRel$ and
$\aRel^{-1}$ are ground simulations.
\end{definition}

\begin{proposition}\label{prop:ground}
Suppose that a relation $\aRel$
\begin{enumerate}
\item is a  ground bisimulation, and
\item is closed under all super-operator application.
\end{enumerate}
Then $\aRel$ is an  open bisimulation.
\end{proposition}
\begin{proof}
Suppose that $\CC\aRel\CD$. Since $\aRel$ is a ground bisimulation, we
have $\qv(\CC)=\qv(\CD)$ and $\ptr(\CC)=\ptr(\CD)$. 
 Since $\aRel$ is closed under all super-operator application, we have $\CE(\CC) \aRel \CE(\CD)$ for any $\CE\in \so(\CH_{\overline{qv(\CC)}})$. If $\CE(\CC)\ar{\alpha}\Delta$, then there exists some distribution $\Theta$ such that $\CE(\CD)\dar{\hat{\alpha}}\Theta$ and $\Delta\lift{\aRel}\Theta$ because $\aRel$ is a ground bisimulation. Similarly, any transtion from $\CE(\CD)$ can also be matched up by $\CE(\CC)$. Therefore, $\aRel$ is an open bisimulation.
%\qed\hfill
\end{proof}
The above proposition provides us a useful proof technique: in order to show that two configurations $\CC$ and $\CD$ are open bisimilar, it sufficies to exhibit a binary relation including the pair $(\CC,\CD)$, and then to check that the relation is a ground bisimulation and is closed under all super-operator application. This is analogous to a proof technique of open bisimulation for the $\pi$-calculus \cite{San96}, where name instantiation is playing the same role as super-operator application here.

\begin{proposition}\label{prop:obisi.so}
$\obisi$ is the largest  ground bisimulation that is closed under all super-operator application.
\end{proposition}
\begin{proof}
By definition $\obisi$ is closed under all super-operator application. It is is obviously a ground bisimulation. Moreover, it is the largest one because of Proposition~\ref{prop:ground}.
\end{proof}

%%%%%%%%%%%%%%%%%%%%%%
\subsection{Equivalence and congruence}
As a sanity check, we prove that $\obisi$
is an equivalence relation. 
This is based on the following transfer property.

\begin{proposition}\label{prop:sbisi.transfer}
  Suppose $\Delta \lift{\obisi} \Theta$ and $\Delta \ar{\alpha} \Delta'$
in a pLTS.
  Then there exists some distribution $\Theta'$ such that
$\Theta \dar{\hat{\alpha}} \Theta'$ and $\Delta' \lift{\obisi} \Theta'$.
\end{proposition}
\begin{proof}
  Suppose $\Delta \lift{\obisi} \Theta$ and $\Delta \ar{\alpha}
  \Delta'$. By Corollary~\ref{cor:lifting.states}, there is a finite index set $I$
  such that (i) $\Delta=\sum_{i\in I}p_i\cdot \pdist{\CC_i}$,
  (ii) $\Theta=\sum_{i\in I}p_i\cdot \pdist{\CD_i}$, and (iii) $\CC_i\obisi\CD_i$ for
  each $i\in I$. By the condition $\Delta\ar{\alpha}\Delta'$, (i) and Proposition~\ref{prop:act.dstruct}, we can decompose
  $\Delta'$ into $\sum_{i\in I}p_i\cdot\Delta'_i$ for some $\Delta'_i$
  such that $\pdist{\CC_i}\ar{\alpha}\Delta'_i$. By
  Lemma~\ref{lem:lifting.alt} again, for each $i\in I$, there is an
  index set $J_i$ such that $\Delta'_i=\sum_{j\in
    J_i}q_{ij}\cdot\Delta'_{ij}$ and $\CC_i\ar{\alpha}\Delta'_{ij}$ for
  each $j\in J_i$ and $\sum_{j\in J_i}q_{ij}=1$. By (iii) there is some $\Theta'_{ij}$ such that
  $\CD_i\dar{\hat{\alpha}}\Theta'_{ij}$ and
  $\Delta'_{ij}\lift{\obisi}\Theta'_{ij}$. Let $\Theta'=\sum_{i\in
    I,j\in J_i}p_iq_{ij}\cdot\Theta'_{ij}$. Since $\dar{\hat{\alpha}}$ is
  linear by Proposition~\ref{prop:linear.move}, we know that
  $\Theta=\sum_{i\in I}p_i\sum_{j\in J_i}q_{ij}\CD_i
  \dar{\hat{\alpha}}\Theta'$. By the linearity of $\lift{\obisi}$, we notice
  that $\Delta'=(\sum_{i\in I}p_i\sum_{j\in J_i}q_{ij}\cdot\Delta'_{ij})
  \lift{\obisi}\Theta'$.
%\hfill\qed
\end{proof}

\begin{corollary}\label{cor:tran.prev}
Suppose 
$\Delta \lift{\obisi} \Theta$ and $\Delta \dar{\hat{\alpha}} \Delta'$.
Then there is some $\Theta'$ with
  $\Theta \dar{\hat{\alpha}} \Theta'$ and $\Delta' \lift{\obisi} \Theta'$.
\end{corollary}
\begin{proof}
By Proposition~\ref{prop:sbisi.transfer}  it is not difficult to show
that

\emph{(*) If $\Delta \lift{\obisi} \Theta$ and $\Delta \dar{\hat{\tau}} \Delta'$ then
  there is some $\Theta'$ with $\Theta \dar{\hat{\tau}} \Theta'$ and
  $\Delta' \lift{\obisi} \Theta'$.}

%  Given the two previous results  this  is fairly straightforward.
  Suppose $\Delta \dar{\alpha} \Delta'$ and $\Delta \;\lift{\obisi}\;
  \Theta$.  If $\alpha$ is $\tau$ then the required $\Theta'$ follows by
  an application of property (*).  
Otherwise, by definition we know $\Delta \dar{\hat{\tau}} \Delta_1,\; \Delta_1
\ar{\alpha} \Delta_2$ and $\Delta_2 \dar{\hat{\tau}} \Delta'$. An application of property
(*) gives a $\Theta_1$ such that $\Theta \dar{\hat{\tau}} \Theta_1$
and $\Delta_1 \lift{\obisi} \Theta_1$. An application of Proposition~\ref{prop:sbisi.transfer} gives
a $\Theta_2$ such that $\Theta_1 \dar{\hat{\alpha}} \Theta_2$ and 
$\Delta_2 \lift{\obisi} \Theta_2$. Finally another application of
property (*)
gives $\Theta_2 \dar{\hat{\tau}} \Theta'$ such that $\Delta' \lift{\obisi} \Theta'$.
The result now follows from the transitivity of $\dar{\hat{\tau}}$.
%\hfill\qed
\end{proof}

\begin{theorem}\label{thm:bisim.equiv}
$\obisi$ is an equivalence relation.
\end{theorem}
\begin{proof}
It is trivial to see that $\obisi$ is reflexive and symmetric. For transitivity, we show that $\obisi\cdot\obisi$ is an open bisimulation relation. Since this is a symmetric relation,  we only need to show that it is an open simulation.
Suppose $\CC_1\obisi\CC_2$ and $\CC_2\obisi\CC_3$.
If $\CC_1\ar{\alpha}\Delta_1$, then there is some $\CC_2\dar{\hat{\alpha}}\Delta_2$ such that $\Delta_1\lift{\obisi}\Delta_2$, since $\CC_1\obisi\CC_2$.  From the condition $\CC_2\obisi\CC_3$ and Corollary~\ref{cor:tran.prev} it follows that $\CC_3\dar{\hat{\alpha}}\Delta_3$ and $\Delta_2\lift{\obisi}\Delta_3$. By Proposition~\ref{prop:lift2} we see that
$\Delta_1\lift{(\obisi\cdot\obisi)}\Delta_3$ as required.
%\hfill\qed
\end{proof}

As a relation between configurations, $\obisi$ is preserved by all static constructors. 

%\begin{lemma}\label{lem:ptr.ext}
%If $S\subseteq S'$ and $\ptr_S(\rho)=\ptr_S(\sigma)$ then
%$\ptr_{S'}(\rho)=\ptr_{S'}(\sigma)$.
%\end{lemma}

\begin{proposition} \label{prop:ocongruence} If 
$\<P,\rho\>\obis \<Q,\sigma\>$ then
\begin{enumerate}
\item
$\<P\| R,\rho\>\obis \<Q\| R,\sigma\>$;
\item
$\<P[f],\rho\>\obis \<Q[f],\sigma\>$;
\item
$\<P\backslash L,\rho\>\obis \<Q\backslash L,\sigma\>$;
\item
$\<\iif\ b\ \then\ P,\rho\>\obis \<\iif\ b\ \then\ Q,\sigma\>$.
\end{enumerate}
\end{proposition}
\begin{proof} We only prove (1) as an example. Let 
$$\r=\{(\<P\| R, \rho\>, \<Q\| R, \sigma\>) \mid \<P, \rho\>\obis \<Q,\sigma\>\}.$$
It suffices to show that $\r$ is an open bisimulation. 
Suppose $\<P\|R, \rho\>\r \<Q\|R,\sigma\>$ where $\<P,\rho\>\obis
\<Q,\sigma\>$. By the definition of $\obis$ we have that $qv(P)=qv(Q)$ and
$\tr_{qv(P)} (\rho) = \tr_{qv(Q)} (\sigma)$. Thus
$qv(P\|R)=qv(Q\|R)$ and we infer that
\[\tr_{qv(P\|R)} (\rho) \;=\; \tr_{qv(P\|R)\backslash qv(P)}
\tr_{qv(P)}(\rho)
\;=\;\tr_{qv(Q\|R)\backslash qv(Q)}
\tr_{qv(Q)}(\sigma)
\;=\; \tr_{qv(Q\|R)} (\sigma).\]
By Proposition~\ref{prop:obisi.so}, we know that  $\pair{P,\CE(\rho)}
\obisi \pair{Q,\CE(\sigma)}$ for any
$\CE\in\so(\CH_{\overline{qv(P)}})$, from which it follows that  
%$\pair{P,\CE(\rho)} \obisi \pair{Q,\CE(\sigma)}$, and thus 
$\pair{P\Cpar R,\CE(\rho)}
\aRel \pair{Q\Cpar R,\CE(\sigma)}$ for any
$\CE\in\so(\CH_{\overline{qv(P\Cpar R)}})$. In other words, $\aRel$ is
closed under super-operator application. Below we show that it is also
a ground bisimulation.

Suppose$\<P\|R, \rho\> \rto{\alpha}\Delta$ for some $\alpha$ and $\Delta$. There are three cases to consider.
\begin{enumerate}
\item The transition is caused by $R$ solely; that is, $\<R, \rho\> \rto{\alpha} \sum_i p_i \cdot\pdist{\<R_i, \e_i(\rho)\>}$, and $\Delta= \sum_i p_i \cdot\pdist{\<P\|R_i, \e_i(\rho)\>}$. Then $$\<Q\|R, \sigma\> \rto{\alpha}\Theta = \sum_i p_i \cdot\pdist{\<Q\|R_i, \e_i(\sigma)\>}.$$ Furthermore, by Proposition~\ref{prop:obisi.so}, we have $\<P, \e_i(\rho)\>\obis \<Q, \e_i(\sigma)\>$, and then $\Delta\lift{\aRel}\Theta$ by definition.

\item The transition is caused by $P$ solely; that is, $\<P, \rho\>
  \rto{\alpha} \Delta_1 = \sum_i p_i\cdot \pdist{\<P_i, \rho_i\>}$, and
  $\Delta= \sum_i p_i \cdot\pdist{\<P_i\|R, \rho_i\>}$. Since $\<P, \rho\>\obis \<Q, \sigma\>$.
Then $\<Q, \sigma\> \Rto{\hat{\alpha}}\Theta_1$ such that
$\Delta_1\lift{\obis} \Theta_1$.  By
Proposition~\ref{prop:act.dstruct}, we have the decomposition
$\Theta_1=\sum_{i}p_i\cdot\Theta_i$ with $\pdist{\<P_i,\rho_i\>}
\lift{\obis}\Theta_i$ for each $i$. 
So we have 
$$\<Q\|R, \sigma\> \Rto{\hat{\alpha}}\Theta = \sum_i p_i\cdot
\Theta_i\| R$$ and $\Delta\lift{\aRel}\Theta$ by definition.

\item The transition is caused by a communication
between $P$ and $R$. Without loss of generality, we assume that
$$\<P,\rho\>\rto{\qc?q} \<P',\rho\>,\ \ \  \<R,\rho\>\rto{\qc!q}\<R',\rho\>,$$  and $\Delta=\<P'\|R',\rho\>$.
By a simple induction on the rules in Figure~\ref{fig:opsem}, it is easy to see that $\<R, \eta\>\rto{\qc !q} \<R',\eta\>$ for any
$\eta\in \dh$.

From the assumption that $\<P, \rho\>\obis \<Q, \sigma\>$, we have
$$\<Q, \sigma\>\Rto{\qc ?q}\sum_{i\in I} p_i\cdot \pdist{\<Q_i,\sigma_i\>}$$  such that for any $i\in I$, it holds that
$\<P', \rho\>\obis \<Q_i, \sigma_i\>$ and
$$\<Q\|R,\sigma\>\Rto{\tau} \Theta=\sum_{i\in I} p_i \cdot\pdist{\<Q_i\|R',\sigma_i\>}.$$ 
Furthermore, for any $i\in I$,  we have
$$(\<P'\|R',\rho\>,\<Q_i\|R',\sigma_i\>)\in \r$$ by definition. That is, $\Delta\lift{\aRel}\Theta$ as required.
\end{enumerate} 
The symmetric form when $\<Q\|R,\e(\sigma)\>\rto{\alpha}\Theta$ can be similarly proved. 
So $\r$ is a  ground bisimulation on $Con$. It follows from
Proposition~\ref{prop:ground} that $\aRel$ is also an open bisimulation.
\hfill %\qed
\end{proof}

Note that we do not have a counterpart of the above proposition for dynamic constructors such as prefix. As a counterexample, consider the following two configurations taken from \cite{FDY11}:
\[\<P,\rho\>  \qquad\text{and}\qquad \<Q,\rho\>\]
where $P=M_{0,1}[q;x].\nil$ with $M_{0,1}=\lambda_0 [|0\>]+\lambda_1 [|1\>]$ being the 1-qubit measurement according to the computational basis, $Q=I[q].\nil$, and $\rho = [|0\>]_q\otimes \sigma$ with
$\sigma\in \d(\h_{\overline{q}})$. We have that $\<P,\rho\>\obisi \<Q,\rho\>$, but 
 $\<H[q].P,\rho\>\not\obisi \<H[q].Q,\rho\>$ when $H$ is the Hadamard operator.

Nevertheless, as a relation between processes, $\obisi$ is preserved by almost all constructors of qCCS.
\begin{theorem}\label{thm:cong}
The relation $\obisi$ between processes is  preserved by all the constructors of qCCS except for summation.
\end{theorem}
\begin{proof}
Similar to the proof of Theorem 6.17 in \cite{FDY11}, which shows the
congruence property of a notion of weak bisimulation.
%\hfill\qed
\end{proof}

%%%%%%%%%%%%%%%%%%%%%%
\subsection{An extensional equivalence}
We formally define three criteria, namely barb-preservation, reduction-closedness and composionality, in order to judge whether two processes are equivalent. This yields an extensional equivalence that turns out to coincide with open bisimilarity.
\begin{definition}[Barbs]
For $\Delta\in\dist{\Con}$ and $c\in cChan$ let
\[V_c(\Delta)=\sum\sset{\Delta(\CC)\mid \CC\ar{c!v} \mbox{ for some $v$}}. \]
We write 
$\CC\barb{c}{\geq p}$ whenever $\CC\dar{\hat{\tau}}\Delta$ for some
$\Delta$ with $V_c(\Delta)\geq p$.
\end{definition}

%If $\CC=\pair{P,\rho}$, we write $\CC||R$ for the configuration $\pair{P||R,\rho}$. If $\Delta=\sum_{i\in I}p_i\cdot\CC_i$, we write $\Delta||R$ for the distribution $\sum_{i\in I}p_i\cdot(\CC_i||R)$.

\begin{definition}
A relation $\aRel$ is 
\begin{itemize}
\item \emph{barb-preserving} if $\CC\aRel
\CD$ implies that $\CC\barb{c}{\geq p}$ iff $\CD\barb{c}{\geq p}$ for
  any classical channel $c$;
\item \emph{reduction-closed} if $\CC\aRel
\CD$ implies
\begin{itemize}
\item whenever $\CC\dar{\hat{\tau}}\Delta$, there exists $\Theta$
  such that $\CD\dar{\hat{\tau}}\Theta$ and
  $\Delta \lift{\aRel}\Theta$,
\item whenever $\CD\dar{\hat{\tau}}\Theta$, there exists $\Delta$
  such that $\CC\dar{\hat{\tau}}\Delta$ and
  $\Delta \lift{\aRel}\Theta$;
\end{itemize}
\item \emph{compositional} if
$\CC \aRel \CD$
                                                                       implies $(\CC||R) \aRel (\CD||R)$ for any process $R$ with $qv(R)$ disjoint from $qv(\CC)\cup qv(\CD)$, and 
$\aRel$ is closed under super-operator application.
\end{itemize}
\end{definition}

\begin{definition}[Reduction barbed congruence]
Let \emph{reduction barbed congruence}, denoted by $\rbc$, be the largest
relation over configurations which is barb-preserving,
reduction-closed and compositional, and furthermore, if $\CC \rbc \CD$ then $\qv(\CC)=\qv(\CD)$ and $\ptr(\CC)=\ptr(\CD)$.
\end{definition}

\begin{theorem}[Soundness]\label{thm:sound}
%For finitary processes,
If $\CC\obisi \CD$ then $\CC \rbc \CD$.
\end{theorem}
\begin{proof}
By Corollary~\ref{cor:tran.prev} and Proposition~\ref{prop:ocongruence}
we know that $\obisi$ is reduction closed and compositional. It remains to show that $\obisi$
is barb-preserving.

Suppose $\CC\obisi\CD$ and $\CC\barb{c}{\geq p}$, for any classical channel $c$ and probability $p$; we need to show that $\CD\barb{c}{\geq p}$. We see from $\CC\barb{c}{\geq p}$ that $\CC\dar{\hat{\tau}}\Delta$ for some $\Delta$ with $V_c(\Delta)\geq p$. By Corollary~\ref{cor:tran.prev}, the relation $\obisi$ is reduction-closed. Hence, there exists $\Theta$ such that $\CD\dar{\hat{\tau}}\Theta$ and $\Delta\lift{\obisi}\Theta$. The latter means that 
\begin{equation}\label{eq:sound1}
\Delta=\sum_{i\in I}p_i\cdot\ \pdist{\CC_i} \qquad
\CC_i \obisi \CD_i \qquad
\Theta=\sum_{i\in I}p_i\cdot\ \pdist{\CD_i}
\end{equation}
By the second part of (\ref{eq:sound1}), if $\CC_i \ar{c!v}$ for some
action $c!v$, then $\CD_i\dar{c!v}$, that is
$\CD_i\dar{\hat{\tau}}\Theta_i\ar{c!v}$ for some distribution $\Theta_i$. Let
$I_c$ be the index set $\sset{i\in I \mid \CC_i\ar{c!v} \mbox{ for some
    $v$}}$, and $\Theta'$ be the distribution
\[(\sum_{i\in I_c} p_i\cdot \Theta_i) +(\sum_{i\in I\backslash
I_c} p_i\cdot\pdist{\CD_i}).\]
By the linearity and reflexivity of $\dar{\hat{\tau}}$,
Proposition~\ref{prop:linear.move}, we have $\Theta\dar{\hat{\tau}}\Theta'$. It follows from $\CD\dar{\hat{\tau}}\Theta\dar{\hat{\tau}}\Theta'$ that
$\CD\dar{\hat{\tau}}\Theta'$. It remains to show that $V_c(\Theta')\geq p$.

Note that for each $i\in I_c$ we have $\Theta_i\ar{c!v}$ for some
action $c!v$, which means that $V_c(\Theta_i)=1$. It follows that
\[\begin{array}{rcl}
V_c(\Theta') & = & \sum_{i\in I_c} p_i\cdot V_c(\Theta_i) +\sum_{i\in I\backslash
I_c} p_i\cdot V_c(\pdist{\CD_i})\\
& \geq &  \sum_{i\in I_c} p_i\cdot V_c(\Theta_i) \\
& = & \sum_{i\in I_c} p_i \\
& = & V_c(\Delta)\\
& \geq & p
\end{array}\]
%\hfill\qed
\end{proof}

In order to obtain completeness, the converse of Theorem~\ref{thm:sound}, we make use of a proof technique that involves examing the barbs of processes in certain contexts; the
following technical lemma enhances this technique.
\begin{lemma}\label{lem:fresh}
If $\Delta||c!0 \lift{\rbc}\Theta||c!0$ where $c$ is a fresh classical
channel,
then $\Delta \lift{\rbc} \Theta$.
\end{lemma}
\begin{proof}
Consider the relation 
\[\aRel = \sset{(\CC,\CD) \mid \CC||c!0 \rbc \CD||c!0 \mbox{ for
    some  fresh channel $c$}}\]
We show that $\aRel\subseteq\rbc$. Suppose $\CC \aRel \CD$. Then
there is a fresh channel $c$ such that
$\CC||c!0 \rbc \CD||c!0$. Let $\CC=\<P,\rho\>$ and $\CD=\<Q,\sigma\>$.
By the definition of $\rbc$ we have 
$\qv(P||c!0)=\qv(Q||c!0)$ and $\ptr(\CC||c!0)=\ptr(\CD||c!0)$,
i.e. $\tr_{\qv(P||c!0)}(\rho)=\tr_{\qv(Q||c!0)}(\sigma)$.
Notice that 
\[
\qv(P) \;=\; \qv(P||c!0) \;=\; \qv(Q||c!0) \;=\; \qv(Q).
\]
It follows that
\[
\ptr(\CC)=\tr_{\qv(P)}(\rho) = \tr_{\qv(P||c!0)}(\rho)
=\tr_{\qv(Q||c!0)}(\sigma) = \ptr(\CD).
\]
Below we check that $\aRel$ is compositional, barb-preserving and reduction-closed.
\begin{enumerate}
\item $\aRel$ is compositional. For any process $R$ with $qv(R)$ disjoint from $qv(\CC)$ and $c$ fresh for $R$, since $\rbc$ is
  compositional, we have $(\CC||c!0||R) \rbc (\CD||c!0||R)$, which
  means $(\CC||R) \aRel (\CD||R)$.
By the compositionality of $\rbc$ we also have 
$\CE(\CC \Cpar c!0) \rbc \CE(\CD \Cpar c!0)$ for any $\CE\in\so(\CH_{\overline{qv(\CC\Cpar c!0)}})$. Since $qv(\CC) = qv(\CC \Cpar c!0)$, we have $\CE\in\so(\CH_{\overline{qv(\CC)}})$. Note that 
\[\CE(\CC)\Cpar c!0 \;=\; \CE(\CC\Cpar c!0)\;\rbc\; \CE(\CD\Cpar c!0) \;=\; \CE(\CD)\Cpar c!0.\] It follows that $\CE(\CC) \aRel \CE(\CD)$ and thus $\aRel$ is closed under super-operator application.

\item $\aRel$ is barb-preserving. Suppose $\CC\barb{c_1}{\geq p}$ for
  some channel $c_1$ and probability $p$. Let $c_2$ be some fresh
  channel. We construct the process $T$ by letting
\[T ~=~ c_1?x_1.c?x.c_2!0\]
for any $x_1$ and $x$. Since $\rbc$ is compositional, we have 
$(\CC \Cpar c!0 \Cpar T) \rbc (\CD \Cpar c!0 \Cpar T)$. Note that $(\CC \Cpar c!0 \Cpar T)\barb{c_2}{\geq p}$,
which implies $(\CD \Cpar c!0 \Cpar T)\barb{c_2}{\geq p}$. Since $c$ is fresh for
$\CD$, the latter has no potential to communicate at channel
$c$. Therefore, it must be the case that $\CD\barb{c_1}{\geq p}$.

\item $\aRel$ is reduction-closed. Suppose $\CC\dar{\hat{\tau}}\Delta$ for some
  distribution $\Delta$. Then $\CC||c!0 \dar{\hat{\tau}} \Delta||c!0$. Since
  $(\CC||c!0) \rbc (\CD||c!0)$, there is some $\Gamma$ such that
  $\CD||c!0 \dar{\hat{\tau}} \Gamma$ and $(\Delta||c!0)\lift{\rbc}\Gamma$. Note that
  $c$ is fresh for $\CD$, thus there is no communication between $\CD$
  and $c!0$. Therefore, it must be the case that $\Gamma=\Theta||c!0$
  for some $\Theta$ such that
$\CD\dar{\hat{\tau}}\Theta$. Thus, $(\Delta||c!0)\lift{\rbc} (\Theta||c!0)$,
i.e. $\Delta \lift{\aRel}\Theta$.
%\hfill\qed
\end{enumerate}
\end{proof}

\begin{theorem}[Completeness]
%For finitary processes,
If $\CC \rbc \CD$ then $\CC \obisi  \CD$.
\end{theorem}
\begin{proof}

Since $\rbc$ is closed under any super-operator application, by Proposition~\ref{prop:ground} it suffices to show that $\rbc$ is a ground
bisimulation. By the symmetry of $\rbc$, we only need to show that $\rbc$ is a ground simulation.
Suppose $\CC=\pair{P,\rho}$, $\CD=\pair{Q,\sigma}$ and
$\CC\rbc\CD$. By definition, we have $\qv(P)=\qv(Q)$ and
$\tr_{\qv(P)}(\rho)=\tr_{\qv(Q)}(\sigma)$. Suppose
$\CC\ar{\alpha}\Delta$. We distinguish several cases.
\begin{enumerate}
\item $\alpha\equiv\tau$. 
Since $\rbc$ is compositional, we have $(\CC \Cpar c!0) \rbc (\CD\Cpar c!0)$ for some fresh channel $c$.
Since $\rbc$ is reduction-closed, the reduction $\CC\Cpar c!0 \dar{\hat{\tau}}\Delta\Cpar c!0$ is matched by some $\Gamma$ such that $\CD\Cpar c!0 \dar{\hat{\tau}}\Gamma$ and $\Delta \Cpar c!0 \lift{\rbc} \Gamma$.
Since $c$ is fresh, there is no communication between $\CD$ and $c!0$, so it must be the case that $\Gamma$ has the form $\Theta\Cpar c!0$ with $\CD\dar{\hat{\tau}}\Theta$. It follows from Lemma~\ref{lem:fresh} and $\Delta \Cpar c!0 \lift{\rbc} \Theta\Cpar c!0$ that $\Delta\lift{\rbc}\Theta$.

\item $\alpha\equiv c!v$. Let $T$ be the process defined by
\[T~:=~ c_1!0+c?x.\ifthen{x=v}{(c_2!0 + \tau.c_3!0)}\]
  where $c_1, c_2$ and $c_3$ are fresh channels. 
Then
  $\CC||T\dar{\hat{\tau}}\Delta||c_3!v$. 
  
  Since $\CC\rbc\CD$ we know
  $\CC||T\rbc\CD||T$ by the compositionality of $\rbc$. Since
  $\rbc$ is reduction-closed, there is some $\Gamma$ such that
  $\CD||T\dar{\hat{\tau}}\Gamma$ and $\Delta||c_3!0 \lift{\rbc} \Gamma$. Since
  $\rbc$ is barb-preserving we have $\Gamma\not\!\barb{c_1}{>0}$, $\Gamma\not\!\barb{c_2}{>0}$
   and
  $\Gamma\barb{c_2}{\geq 1}$. Here we use the notation $\Gamma\not\!\barb{c_1}{>0}$ to
    mean that $\Gamma\barb{c_1}{\geq p}$ does not hold for any $p>0$.
It must be the case that $\Gamma\equiv
  \Theta||c_{3}!0$ for some $\Theta$ with $\CD\dar{c!v}\Theta$. By
  Lemma~\ref{lem:fresh} and $\Delta||c_3!0 \lift{\rbc} \Theta||c_3!0$, we have $\Delta\lift{\rbc}\Theta$.

\item $\alpha\equiv \qc ? q$. Let $T$ be the process defined by
\[T~:=~ c_1!0+\qc!r.c_2!0\]
  where $c_1$ and $c_2$ are fresh channels. Then
  $\CC||T\dar{\hat{\tau}}\Delta \Cpar c_2!0$. 
  Since $\CC\rbc\CD$ we know
  $\CC||T\rbc\CD||T$ by the compositionality of $\rbc$. Since
  $\rbc$ is reduction-closed, there is some $\Gamma$ such that
  $\CD||T\dar{\hat{\tau}}\Gamma$ and $\Delta||c_2!0 \rbc \Gamma$. Since
  $\rbc$ is barb-preserving we have $\Gamma\not\!\barb{c_1}{>0}$ and
  $\Gamma\barb{c_2}{\geq 1}$. It follows that
 $\CD\dar{\qc?q}\Theta$ and $\Gamma\equiv \Theta\Cpar c_{2}!0$, with
implicit assumption of $\alpha$-conversion. 
By Lemma~\ref{lem:fresh} and $\Delta\Cpar c_2!0 \lift{\rbc} \Theta\Cpar c_2!0$,
  we have $\Delta \lift{\rbc}\Theta$. 

The case when $\alpha\equiv c?x$ is similar.

\item $\alpha\equiv \qc ! q$. Let $T$ be the process defined by
\[T~:=~ c_1!0 + \qc ? r.(c_2!0 + I[r].c_3!0)\]
  where $c_1,c_2$ and $c_3$ are fresh channels. Then
  $\CC||T\dar{\hat{\tau}}\Delta \Cpar (c_2!0 + I[q].c_3!0)$. 
  Since $\CC\rbc\CD$ we know
  $\CC||T\rbc\CD||T$ by the compositionality of $\rbc$. Since
  $\rbc$ is reduction-closed, there is some $\Gamma$ such that
  $\CD||T\dar{\hat{\tau}}\Gamma$ and 
  \begin{equation}\label{eq:ytemp1}
  \Delta||(c_2!0 + I[q].c_3!0) \lift{\rbc} \Gamma.
  \end{equation}
  Since
  $\rbc$ is barb-preserving we have $\Gamma\not\!\barb{c_1}{> 0}$ and
  $\Gamma\barb{c_2}{\geq 1}$. It follows that
 $\CD\dar{\qc!q'}\Theta$ for some $q'\in qVar$, and $\Gamma\equiv \Theta\Cpar  (c_2!0 + I[q'].c_3!0)$. 
 Note that $\Delta||(c_2!0 + I[q].c_3!0)\ar{\tau}  \Delta||c_3!0$. To match this action, we have $\Gamma\dar{\hat{\tau}}\Gamma'$ for some $\Gamma'$
such that $\Delta||c_3!0\lift{\rbc} \Gamma'$. As a consequence, we have $\Gamma'\barb{c_3}{\geq 1}$ but $\Gamma'\not\!\barb{c_2}{> 0}$, so
$\Gamma'\equiv \Theta'\Cpar  c_3!0$ for some $\Theta'$ with  $\Theta\dar{\hat{\tau}}\Theta'$, which implies $\CD\dar{\qc!q'}\Theta'$. Now by Lemma~\ref{lem:fresh} and $\Delta\Cpar c_3!0 \lift{\rbc} \Theta'\Cpar c_3!0$,
  we derive $\Delta \lift{\rbc}\Theta'$. 
  
 Finally, we claim that $q=q'$. Otherwise from Eq.(\ref{eq:ytemp1}), we know $q'\in qv(\Delta)$ but $q'\not\in qv(\Theta)$.
 That contradicts the fact that $\Delta||c_3!0\lift{\rbc} \Gamma'$ as  $qv(\Theta')\subseteq qv(\Theta)$.

%\hfill\qed
\end{enumerate}
\end{proof}

%%%%%%%%%%%%%%%%%%%%%%%%
\subsection{Modal characterisation}
We extend the Hennessy-Milner logic by adding a probabilistic choice
modality to express the bebaviour of distributions, as in \cite{DGHM09}, and a
super-operator modality to express trace-preserving super-operator
application, as well as atomic formulae involving projectors for dealing with density operators.

\begin{definition}\label{def:logic}
The class $\CL$ of modal formulae over $\Act$, ranged over by
$\phi$, is defined by the following grammar:
\[\begin{array}{rcl}
\phi & := & E_{\tilde{q}}^{\geq p}\mid \bigwedge_{i\in I}\phi_i \mid \diam{\alpha}\psi \mid
\neg\phi \mid \CE.\phi\\
\psi &:=& \bigoplus_{i\in I}p_i\cdot\phi_i \end{array}\] 
%where the index set $I$ in $\bigwedge_{i\in I}\phi_i$ can be infinite.
where $\alpha\in \Act_{\tau}$, $\e$ is a super-operator, and $E$ is a projector associated with a certain subspace of $\h_{\widetilde{q}}$.
We call
$\phi$ a \emph{configuration formula} and $\psi$ a \emph{distribution
formula}. Note that a distribution formula $\psi$ only appears as
the continuation of a diamond modality $\diam{\alpha}\psi$. 

The \emph{satisfaction relation}\index{satisfaction relation}
$\models \subseteq S\times\CL $ is defined by
\begin{itemize}
\item $\CC\models E_{\tilde{q}}^{\geq p}$ 
if $qv(\CC)\cap\tilde{q}=\emptyset$ and $\tr(E_{\tilde{q}}\rho)\geq p$ where $\CC=\pair{P,\rho}$.
\item $\CC\models \bigwedge_{i\in I}\phi_i $ if $\CC\models\phi_i$ for all $i\in I$.
\item $\CC\models \diam{\alpha}\psi$ if for some $\Delta\in \dist{\Con}$,
  $\CC\dar{\hat{\alpha}}\Delta$ and $\Delta\models\psi$.
\item $\CC\models\neg\phi$ if it is not the case that
  $\CC\models\phi$.
\item $\CC\models\CE.\phi$ if $\CE\in\so(\CH_{\overline{qv(\CC)}})$ and $\CE(\CC)\models\phi$.
\item $\Delta\models\bigoplus_{i\in I}p_i\cdot\phi_i$ if there are
  $\Delta_i\in\dist{\Con}$, for all $i\in I, \CD\in\support{\Delta_i}$, with
  $\CD\models\phi_i$, such that $\Delta=\sum_{i\in I}p_i\cdot\Delta_i$.
\end{itemize}
\end{definition}
With a slight abuse of notation, we write $\Delta\models\psi$ above
to mean that $\Delta$ satisfies the distribution formula $\psi$.
A logical equivalence arises from the logic naturally: we write $\CC
=^\CL \CD$ if $\CC\models\phi \Leftrightarrow \CD\models\phi$ for all
$\phi\in\CL$.

It turns out that $\CL$ is adequate with respect to open
bisimilarity.
\begin{theorem}%[Adequacy]
\label{p:modal.characterisation}
Let $\CC$ and $\CD$ be any two configurations in a pLTS. Then
$\CC\obisi \CD$ if and only if $\CC =^\CL \CD$.
\end{theorem}
\begin{proof}
($\Rightarrow$)
 Suppose $\CC \obisi \CD$, we show that $\CC\models\phi \Leftrightarrow
\CD\models\phi$. Since $\obisi$ is symmetric, it suffices to prove
that $\CC\models\phi$ implies $\CD\models\phi$
 by structural induction on $\phi$.
\begin{itemize}
\item Let $\CC\models E_{\tilde{q}}^{\geq p}$.  Then $qv(\CC)\cap
  \tilde{q}=\emptyset$ and $\tr(E_{\tilde{q}}\rho)\geq p$. Since
  $\CC\obisi \CD$, we have $qv(\CC)=qv(\CD)$ and
  $\ptr(\CC)=\ptr(\CD)$. Thus $qv(\CD)\cap \tilde{q}=\emptyset$. Let
  $\CC=\pair{P,\rho}$ and $\CD=\pair{Q,\sigma}$. We can infer that
\[\begin{array}{rcl}
\tr(E_{\tilde{q}}\sigma) & = & \tr_{\overline{qv(Q)}}\tr_{qv(Q)}(E_{\tilde{q}}\sigma)\\
& = & \tr_{\overline{qv(Q)}}E_{\tilde{q}}(\tr_{qv(Q)}(\sigma))\\ 
& = & \tr_{\overline{qv(P)}}E_{\tilde{q}}(\tr_{qv(P)}(\rho))\\ 
& = & \tr_{\overline{qv(P)}}\tr_{qv(P)}(E_{\tilde{q}}\rho)\\
& = & \tr(E_{\tilde{q}}\rho)\\
& \geq & p .
\end{array}\]
It follows that
$\CD\models E_{\tilde{q}}^{\geq p}$.
\item Let $\CC \models\bigwedge_{i\in I}\phi_i$. Then $\CC\models\phi_i$
  for each $i\in I$. So by induction $\CD\models\phi_i$, and we have
  $\CD\models \bigwedge_{i\in I}\phi_i$. 
%By symmetry we also have $\CD\models \bigwedge_{i\in I}\phi_i$ implies $\CC\models \bigwedge_{i\in I}\phi_i$.
\item Let $\CC\models\neg\phi$. So $\CC\not\models\phi$, and by induction
  we have $\CD\not\models\phi$. Thus $\CD\models\neg\phi$. 
%By symmetry we also have $\CD\not\models\phi$ implies $\CC\not\models\phi$.
\item Let $\CC\models\diam{\alpha}\bigoplus_{i\in I}p_i\cdot\phi_i$. Then $\CC\dar{\hat{\alpha}}\Delta$ and $\Delta\models\bigoplus_{i\in I}p_i\cdot\phi_i$
  for some $\Delta$. So
  $\Delta=\sum_{i\in i}p_i\cdot\Delta_i$ and for all $i\in I$ and
  $\CC'\in\support{\Delta_i}$ we have $\CC'\models\phi_i$. Since
  $\CC\obisi \CD$, by Corollary~\ref{cor:tran.prev} there is some $\Theta$ with $\CD\dar{\hat{\alpha}}\Theta$
  and $\Delta\lift{\obisi} \Theta$. Since the lifted relation is left-decomposable,  we have that $\Theta=\sum_{i\in
  I}p_i\cdot\Theta_i$ and $\Delta_i\lift{\obisi}\Theta_i$. It follows
  that for each $\CD'\in\support{\Theta_i}$ there is some
  $\CC'\in\support{\Delta_i}$ with $\CC'\obisi \CD'$.
  So by induction we have $\CD'\models\phi_i$ for all $\CD'\in\support{\Theta_i}$ with $i\in I$.
  Therefore, we have $\Theta\models\bigoplus_{i\in
  I}p_i\cdot\phi_i$. It follows that $\CD\models\diam{\alpha}\bigoplus_{i\in
  I}p_i\cdot\phi_i$.
  %By symmetry we also have $\CD\models\diam{\alpha}\bigoplus_{i\in I}p_i\cdot\phi_i\Rightarrow \CC\models\diam{\alpha}\bigoplus_{i\in I}p_i\cdot\phi_i$.

\item Let $\CC\models \CE.\phi$. Then
  $\CE\in\so(\CH_{\overline{qv(\CC)}})$ and
  $\CE(\CC)\models\phi$. Since $\CC\obisi\CD$, we have
  $\CE(\CC) \obisi \CE(\CD)$ by Proposition~\ref{prop:obisi.so} and $qv(\CC) =
  qv(\CD)$. By induction, we have $\CE(\CD)\models\phi$. It follows
  that $\CD\models\CE.\phi$. 
%By symmetry we also have $\CD\models\CE.\phi$ implies $\CC\models\CE.\phi$.
\end{itemize}

($\Leftarrow$)  Suppose $\CC =^\CL \CD$.
We first show that $qv(\CC)=qv(\CD)$ and $\ptr(\CC)=\ptr(\CD)$.
For any $\tilde{q}$, if $\tilde{q}\cap qv(\CC)=\emptyset$ then
$\CC\models I^{\geq 1}_{\tilde{q}}$. Since $\CC =^\CL \CD$ we have
$\CD\models I^{\geq 1}_{\tilde{q}}$, and thus $\tilde{q}\cap
qv(\CD)=\emptyset$. It follows that $\qv(\CC)\supseteq \qv(\CD)$. By
the symmetry of $=^\CL$, this implies $qv(\CC)=qv(\CD)$.
Now let $\CC=\pair{P,\rho}$ and $\CD=\pair{Q,\sigma}$. 
Suppose for a contradiction that $\tr_{qv(P)}\rho \not=
\tr_{qv(P)}\sigma$. Then there exists a projection $E$ on
$\tilde{q}$ with $\tilde{q}\cap qv(P)=\emptyset$ and
$\tr(E_{\tilde{q}}\sigma) < \tr(E_{\tilde{q}}\rho)$. Let
$p=\tr(E_{\tilde{q}}\rho)$. Then $\pair{P,\rho}\models
E_{\tilde{q}}^{\geq p}$ while $\pair{Q,\sigma}\not\models
E_{\tilde{q}}^{\geq p}$, contradicting the assumption that $\CC =^\CL \CD$.

Next, we show that the relation $=^\CL$ is a ground
bisimulation. Suppose $\CC =^\CL \CD$ and $\CC\ar{\alpha}\Delta$. We have to
show that there is some $\Theta$ with $\CD\dar{\hat{\alpha}}\Theta$ and $\Delta
\lift{(=^\CL)} \Theta$. Consider the set
\begin{equation}\label{eq:T}
T:=\{\Theta \mid \CD\dar{\hat{\alpha}}\Theta \wedge \Theta=\sum_{\CC'\in\support{\Delta}}\Delta(\CC')\cdot \Theta
_{\CC'}\wedge \exists \CC'\in\support{\Delta},\exists
\CD'\in\support{\Theta_{\CC'}}: \CC'\not=^\CL \CD'\}
\end{equation}
 For each $\Theta\in T$, there must be some
 $\CC'_\Theta\in\support{\Delta}$ and
 $\CD'_\Theta\in\support{\Theta_{\CC'_\Theta}}$ such that (i) either there is a formula $\phi_{\Theta}$ with
$\CC'_\Theta\models\phi_{\Theta}$ but
$\CD'_\Theta\not\models\phi_{\Theta}$ (ii) or there is a formula
$\phi'_{\Theta}$ with $\CD'_\Theta\models\phi'_{\Theta}$ but
$\CC'_\Theta\not\models\phi'_{\Theta}$. In the latter case we set
$\phi_{\Theta}=\neg\phi'_{\Theta}$ and return back to the former
case. So for each $\CC'\in\support{\Delta}$ it holds that
$\CC'\models\bigwedge_{\sset{\Theta\in T\mid \CC'_\Theta =
\CC'}}\phi_\Theta$ and for each $\Theta\in T$ with $\CC'_\Theta=\CC'$
there is some $\CD'_{\Theta}\in\support{\Theta_{\CC'}}$ with
$\CD'_{\Theta}\not\models \bigwedge_{\sset{\Theta\in T\mid \CC'_\Theta =
\CC'}}\phi_\Theta$. Let
\begin{equation}\label{eq:phi}
\phi:=\diam{\alpha}\bigoplus_{\CC'\in\support{\Delta}}\Delta(\CC')\cdot\bigwedge_{\sset{\Theta\in
    T\mid \CC'_{\Theta}=\CC'}}\phi_{\Theta}.
\end{equation}
It is clear that $\CC\models\phi$, hence $\CD\models\phi$ by $\CC=^\CL \CD$.
It follows that there must be a $\Theta^\ast$ with
$\CD\dar{\hat{\alpha}}\Theta^\ast$,
$\Theta^\ast=\sum_{\CC'\in\support{\Delta}}\Delta(\CC')\cdot\Theta^\ast_{\CC'}$
and for each $\CC'\in\support{\Delta},
\CD'\in\support{\Theta^\ast_{\CC'}}$ we have
$\CD'\models\bigwedge_{\sset{\Theta\in T\mid \CC'_\Theta =
\CC'}}\phi_\Theta$. This means that $\Theta^\ast\not\in T$ and hence
for each $\CC'\in\support{\Delta}, \CD'\in\support{\Theta^\ast_{\CC'}}$ we
have $\CC'=^\CL \CD'$. It follows that $\Delta\lift{(=^\CL)}
\Theta^\ast$. By symmetry all transitions of $\CD$ can be matched up
by transitions of $\CC$.

Finally, we prove that the relation $=^{\CL}$ is closed under
super-operator application. That is, for any
$\CE\in\so(\CH_{\overline{qv(\CC)}})$ we need to show that $\CC =^\CL
\CD$ implies
$\CE(\CC)
=^\CL \CE(\CD)$. Suppose $\CC =^\CL
\CD$ and let $\phi$ be any formula such that
$\CE(\CC)\models\phi$. We have $\CC\models\CE.\phi$. It follows from
$\CC =^\CL \CD$ that $qv(\CC)=\qv(\CD)$ and $\CD\models
\CE.\phi$. Therefore, we obtain $\CE(\CD)\models\phi$. By symmetry if
$\phi$ is satisfied by $\CE(\CD)$ then it is also satisfied by
$\CE(\CC)$. In other words, we have $\CE(\CC)
=^\CL \CE(\CD)$.

Now by appealing to Proposition~\ref{prop:ground} we see that $=^\CL$
is an open bisimulation, thus $=^\CL \;\subseteq\; \obisi$.
%\hfill\qed
\end{proof}
Note that the set $T$ in (\ref{eq:T}) is infinite in general as $\CD$
may have infinitely many different derivatives, hence we
have to use infinite conjunction in (\ref{eq:phi}). This is the reason
that we cannot restrict ourselves to finite or binary conjunction in Definition~\ref{def:logic}. 
%%%%%%%%%%%%%%%
\section{Examples}
\label{sec:examples}
\newcommand {\qcf}[1] {{\sf{#1}}}
\newcommand{\con}[3]{\iif\ {#1}\ \then\ {#2}\ \eelse\ {#3}}
%\newcommand{\define}{\stackrel{def}=}

%\begin{document}%

BB84, the first quantum key distribution protocol developed by Bennett and Brassard in 1984 \cite{BB84}, provides a provably secure way to
create a private key between two parties, say, Alice and Bob.
Its security relies on the basic property of quantum mechanics that information gain about a quantum state is only possible at the expense of changing the state, if the  states to be distinguished are not orthogonal.
The basic BB84 protocol goes as follows:
\begin{enumerate}
\item[(1)] Alice randomly creates two strings of bits $\tilde{B}_a$ and  $\tilde{K}_a$, each with size $n$.
\item[(2)] Alice prepares a string of qubits $\tilde{q}$, with size $n$, such that 
the $i$th qubit of $\tilde{q}$ is $|x_y\>$ where $x$ and $y$ are the $i$th bits of $\tilde{B}_a$ and  $\tilde{K}_a$, respectively,
and $|0_0\> = |0\>$, $|0_1\> = |1\>$,
$|1_0\> = |+\>$, and $|1_1\> = |-\>$. Here the symbols $|+\>$ and $|-\>$ have their usual meaning:
\[|+\> = (|0\> + |1\>)/\sqrt{2} \qquad\mbox{and} \qquad 
|-\> = (|0\> - |1\>)/\sqrt{2}. \]
\item[(3)] Alice sends the qubit string $\tilde{q}$ to Bob.
\item[(4)] Bob randomly generates a string of bits $\tilde{B}_b$ with size $n$.
\item[(5)] Bob measures each qubit received from Alice according to a basis determined by the bits he generated: if the $i$th bit of $\tilde{B}_b$ is $k$ then
he measures with $\{|k_0\>, |k_1\>\}$, $k=0,1$. Let the measurement results be $\tilde{K}_b$, which is also a string of bits with size $n$.
\item[(6)] Bob sends his choice of measurement bases $\tilde{B}_b$ back to Alice, and upon receiving the information, Alice sends her bases  
$\tilde{B}_a$ to Bob.
\item[(7)] Alice and Bob determine at which positions the bit strings $\tilde{B}_a$ and  $\tilde{B}_b$ are equal. They discard the bits in 
$\tilde{K}_a$ and $\tilde{K}_b$ where the corresponding bits of $\tilde{B}_a$ and  $\tilde{B}_b$ do not match.
\end{enumerate}
After the execution of the basic BB84 protocol above, the remaining bits of $\tilde{K}_a$ and $\tilde{K}_b$, denoted by $\tilde{K}'_a$ and $\tilde{K}'_b$ respectively, should be the same, provided that the channels used are perfect, and no eavesdropper exists. 

To detect a potentially existing eavesdropper Eve, Alice and Bob proceed as follows:

\begin{enumerate}
\item[(8)] Alice randomly chooses $\support{k/2}$, where $k$ is the size of $\tilde{K}'_a$, bits of $\tilde{K}'_a$, denoted by $\tilde{K}''_a$, and sends
Bob $\tilde{K}''_a$ and their indexes in the original string $\tilde{K}'_a$.

\item[(9)]  Upon receiving the information from Alice, Bob sends back to Alice his substring $\tilde{K}''_b$ of $\tilde{K}'_b$ according to the indexes received from Alice. 

\item[(10)] Alice and Bob check if the strings $\tilde{K}''_a$ and $\tilde{K}''_b$ are equal. If yes, then the remaining substrings $\tilde{K}^f_a$ (resp. $\tilde{K}^f_b$) of $\tilde{K}'_a$ (resp. $\tilde{K}'_b$) by deleting $\tilde{K}''_a$ (resp. $\tilde{K}''_b$) are the secure keys shared by Alice and Bob. Otherwise, an eavesdropper is detected, and the protocol halts without generating any secure keys.
\end{enumerate}

% Replace the misleading 'magic channel' $Ran_{\tilde{q}}$ with a more proper notation of quantum measurement.

For simplicity, we omit the processes of information reconciliation and privacy amplification. Now we describe the above protocol in our language of qCCS.  To ease the notations, we assume a special measurement $Ran[\tilde{q}; \tilde{x}]$ which can create a string of $n$ random bits, independent of the initial states of the $\tilde{q}$ system, and store it to $\tilde{x}$. In effect, $Ran[\tilde{q}; \tilde{x}]=
Set^n_+[\tilde{q}].M^n_{0,1}[\tilde{q};\tilde{x}].Set^n_0[\tilde{q}]$. Then the basic BB84 protocol can be defined as
\begin{eqnarray*}
Alice&\define& Ran[\tilde{q}; \tilde{B}_a].Ran[\tilde{q}; \tilde{K}_a].Set_{\tilde{K}_a}[\tilde{q}].H_{\tilde{B}_a}[\tilde{q}].\qcf{A2B}!\tilde{q}.WaitA(\tilde{B}_a, \tilde{K}_a)
\\
\\
WaitA(\tilde{B}_a, \tilde{K}_a)&\define& b2a?\tilde{B}_b.a2b!\tilde{B}_a.key_a!cmp(\tilde{K}_a, \tilde{B}_a, \tilde{B}_b).\nil\\
\\
Bob&\define& \qcf{A2B}?\tilde{q}.Ran[\tilde{q}'; \tilde{B}_b].M_{\tilde{B}_b}[\tilde{q};\tilde{K}_b].b2a!\tilde{B}_b.WaitB(\tilde{B}_b, \tilde{K}_b)\\
\\
WaitB(\tilde{B}_b, \tilde{K}_b)&\define& a2b?\tilde{B}_a.key_b!cmp(\tilde{K}_b, \tilde{B}_a, \tilde{B}_b).\nil\\
\\
BB84 &\define& (Alice\| Bob)\backslash\{a2b, b2a, \qcf{A2B}\}
\end{eqnarray*}
where $Set^n_+$ is the super-operator which sets each of the $n$ qubits it applies on to $|+\>$,
$M_{\tilde{y}}[\tilde{q}; \tilde{K}_b]$ is the quantum measurement on $\tilde{q}$ according to the basis determined by $\tilde{y}$, i.e., for each $1\leq k\leq n$, it measures $q_k$ with respect to the basis $\{|0\>, |1\>\}$ (reps. $\{|+\>, |-\>\}$) if $y(k)=0$ (resp. 1), and stores the result into $\tilde{K}_b(k)$. $M^n_{0,1}$ is the same as $M_{0\cdots 0}$, and $H_{\tilde{y}}[\tilde{q}]$ has a similar meaning with $M_{\tilde{y}}[\tilde{q}; \tilde{K}_b]$. We also abuse the notion slightly by writing
$\e_{\tilde{B}}[\tilde{q}].P$ when we mean 
$\sum_{\tilde{x}=0^n}^{1^n}(\iif\ \tilde{B}=\tilde{x}\ \then\ \e_{\tilde{x}}[\tilde{q}].P)$ where $0^n$ is the all zero string of size $n$. The function $cmp$ takes a triple of strings $\tilde{x},\tilde{y},\tilde{z}$ with the same size as inputs, and
returns the substring of $\tilde{x}$ where the corresponding bits of $\tilde{y}$ and  $\tilde{z}$ match. When $\tilde{y}$ and  $\tilde{z}$ match nowhere, we let $cmp(\tilde{x}, \tilde{y}, \tilde{z})=\epsilon$, the empty string.

To show the correctness of this basic form of BB84 protocol, we have two choices. The first one is to employ the concept of bisimulation. Let
\begin{eqnarray*}
BB84_{spc} &\define & Ran[\tilde{q}; \tilde{B}_a].Ran[\tilde{q}; \tilde{K}_b].Ran[\tilde{q}'; \tilde{B}_b].\\
&&(key_a!cmp(\tilde{K}_b, \tilde{B}_a, \tilde{B}_b).\nil\| key_b!cmp(\tilde{K}_b, \tilde{B}_a, \tilde{B}_b).\nil).
\end{eqnarray*}
 The pLTSs of $BB84$ and $BB84_{spe}$ for the special case of $n=2$ can be depicted as in Figure~\ref{fig:exm1}, where for simplicity, we only specify the branch where $\tilde{B}_a=\tilde{K}_a=00$. Each arrow in the graph denotes a sequence of $\tau$ actions, and all probabilistic distributions are uniform. The strings at the bottom line are the outputs of the protocol. Then
it can be easily checked from the pLTSs that $BB84\obis BB84_{spe}$. The key is, for each extra branch in $BB84$ caused by the measurement of Bob (the $\tilde{K}_b$ line), the final states are bisimilar; they all output the same string.

\begin{figure}[t]
\[
\begin{array}{ll}
\includegraphics[width=0.5\textwidth]{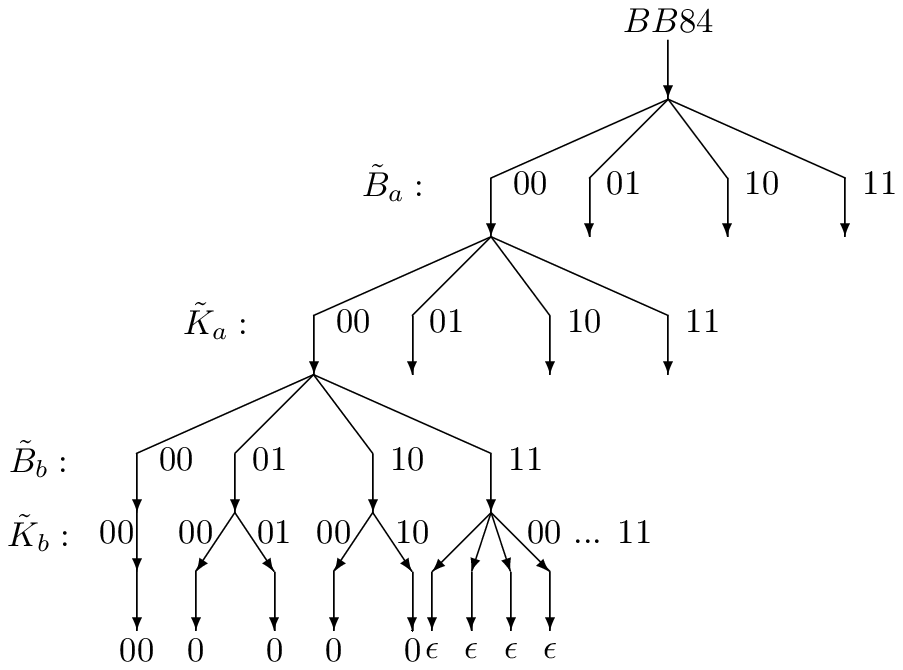}
\includegraphics[width=0.5\textwidth]{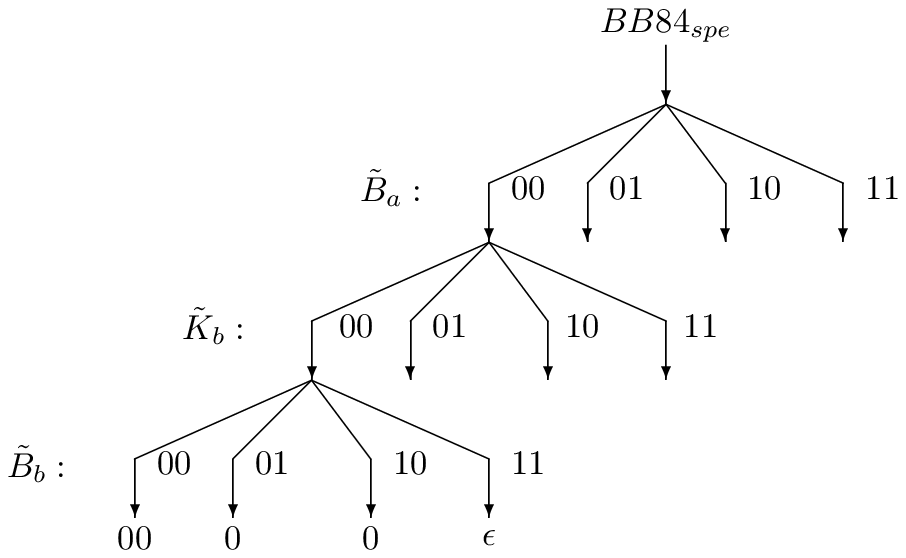}
\end{array}
 \] 
 \caption{pLTSs for $BB84$ and $BB84_{spe}$\label{fig:exm1}}
\end{figure}

The second choice is to use logic formulae. Let
\begin{eqnarray*}
TestBB84&\define& (BB84\| key_a?\tilde{K}'_a.key_b?\tilde{K}'_b.\\
&&(\iif\ \tilde{K}'_a=\tilde{K}'_b\ \then\ suc!0.\nil\ \eelse\ fail!0))\backslash\{key_a, key_b\},
\end{eqnarray*}
and 
$$\psi_p =  \<suc!0\>true \wedge \neg\<\tau\> (p\cdot \<fail!0\>true + (1-p)\cdot true)$$ where $true$ is the abbreviation of $\bigwedge_{i\in \emptyset} \phi_i$. 
It is not difficult to show $TestBB84\models \psi_p$ for any $p>0$.

Now we proceed to describe the protocol where an eavesdropper can be detected. 
\begin{eqnarray*}
Alice'&\define & (Alice\|  key_a?\tilde{K}_a'.Pstr_{|\tilde{K}_a'|}[\tilde{q}_a;\tilde{x}].a2b!\tilde{x}.a2b!SubStr(\tilde{K}_a', \tilde{x}).b2a?\tilde{K}_b''.\\
&&(\iif\ SubStr(\tilde{K}_a', \tilde{x})=\tilde{K}_b''\ \then\ key'_a!RemStr(\tilde{K}_a', \tilde{x}).\nil\\
&&\hspace{11em}\eelse\ alarm_a!0.\nil )))\backslash\{key_a\}\\
\\
Bob'&\define& (Bob\| key_b?\tilde{K}_b'.a2b?\tilde{x}.a2b?\tilde{K}_a''.b2a!SubStr(\tilde{K}_b', \tilde{x}).\\
&&(\iif\ SubStr(\tilde{K}_b', \tilde{x})=\tilde{K}_a''\ \then\ key'_b!RemStr(\tilde{K}_b', \tilde{x}).\nil\\
&&\hspace{11em}\eelse\ alarm_b!0.\nil))\backslash\{key_b\}\\
\\
BB84'&\define& Alice'\| Bob'
\end{eqnarray*}
where $|\tilde{x}|$ is the size of $\tilde{x}$, the function $SubStr(\tilde{K}_a', \tilde{x})$ returns the substring of $\tilde{K}_a'$
at the indexes specified by $\tilde{x}$, and $RemStr(\tilde{K}_a', \tilde{x})$ returns the remaining substring of $\tilde{K}_a'$ by deleting $SubStr(\tilde{K}_a', \tilde{x})$. The special measurement $Pstr_{m}$, which is similar to $Ran$, randomly generates a $\lceil{m/2}\rceil$-sized string of indexes from $1,\dots,m$. 

For the capacity of a potential eavesdropper Eve, we assume that she has complete control of the quantum channel, but can only listen on the classical channels  between Alice and Bob. That is, she can do any quantum operations on the communicated qubits from Alice and Bob, one of the extreme cases being keeping the qubits from Alice while creating and sending to Bob some fresh ones, with the same size, prepared by herself. But for classical communication, Eve can only copy and resend the bits without altering them, since Alice and Bob can choose to send them through a broadcasting channel. Note that perfect copying of the qubits transmitted through the quantum channel from Alice to Bob is prohibited by the basic laws of quantum mechanics, since the potential quantum states sent, $|0\>, |1\>, |+\>$, and $|-\>$ in this protocol, are nonorthogonal. With these natural assumptions, an eavesdropper Eve can be described as:
\begin{eqnarray*}
Eve &\define& \qcf{A2E}? \tilde{q}.\e[\tilde{q}'',\tilde{q}].M^n_{0,1}[\tilde{q}'';\tilde{K}_e].\qcf{E2B}!
\tilde{q}.WaitE(\tilde{K}_e)\\
WaitE(\tilde{K}_e)&\define& b2e?\tilde{B}_b.e2a!\tilde{B}_b.a2e?\tilde{B}_a.e2b!\tilde{B}_a.
a2e?\tilde{x}.e2b!\tilde{x}.\\
&&a2e?\tilde{K}_a''.e2b!\tilde{K}_a''.
b2e?\tilde{K}_b''.e2a!\tilde{K}_b''.
key'_e!gkey(\tilde{K}_e, \tilde{B}_e, \tilde{B}_a, \tilde{B}_b, \tilde{K}_a'', \tilde{K}_b'', \tilde{x}).\nil
\end{eqnarray*}
where $\e$ is a super-operator, and $gkey$ is the function Eve used to generate her guess of the key from the classical information transmitted between  Alice and Bob. Then a practical running BB84 protocol, with the existence of an eavesdropper, goes as follows 
\begin{eqnarray*}
BB84_E &\define&  (Alice'[f_a]\| Eve \| Bob'[f_b])\backslash\{a2e, b2e, e2a, e2b, \qcf{A2E}, \qcf{E2B}\}
\end{eqnarray*}
where $f_a$ and $f_b$ are relabelling functions such that $f_a(a2b)=a2e, f_a(b2a)=e2a,$ $f_a(\qcf{A2B})=\qcf{A2E}$, and $f_b(a2b)=e2b, f_b(b2a)=b2e, f_b(\qcf{A2B})=\qcf{E2B}$.

To get a taste of the security of $BB84'$, we consider a special case where Eve's strategy is to simply measure the qubits sent by Alice, according to randomly guessed  bases, to get the keys. She then prepares and sends to Bob a fresh sequence of qubits, employing the same method Alice used to encode keys, but using her own guess of bases and the keys she obtained. That is, we define
\begin{eqnarray*}
Eve' &\define& \qcf{A2E}? \tilde{q}.Ran[\tilde{q}''; \tilde{B}_e].M_{\tilde{B}_e}[\tilde{q};\tilde{K}_e].Set_{\tilde{K}_e}[\tilde{q}].H_{\tilde{B}_e}[\tilde{q}].\qcf{E2B}! \tilde{q}.WaitE(\tilde{K}_e)
\end{eqnarray*}
Now let $BB84_E'$ be the protocol obtained from $BB84_E$ by replacing Eve by Eve$'$, and letting the function $gkey$ simply return its first parameter. Let
\begin{eqnarray*}
TestBB84'&\define& (BB84'_E\| key'_a?\tilde{x}.key'_b?\tilde{y}.key'_e?\tilde{z}.(\iif\ \tilde{x}\neq \tilde{y}\ \then\ fail!0.\nil\\
&&\hspace{12em}\eelse\ key_e!\tilde{z}.skey!\tilde{x}.\nil))\backslash\{key'_a, key'_b, key'_e\}.
\end{eqnarray*} 
It is generally very complicated to prove the security of the full $BB84$ protocol, even for the simplified $Eve'$ presented above. Here we choose to reduce $TestBB84'$ to a simpler process which is easier for further verification. To be specific, we can show that $TestBB84'$ is bisimilar to the following process:
\begin{eqnarray*}
TB&\define& Ran[\tilde{q}; \tilde{B}_a].Ran[\tilde{q}; \tilde{K}_a].Ran[\tilde{q}''; \tilde{B}_e].
Ran'_{\tilde{B}_a, \tilde{B}_e, \tilde{K}_a}[\tilde{q}; \tilde{K}_e].Ran[\tilde{q}'; \tilde{B}_b].\\
&&Ran'_{\tilde{B}_{e}, \tilde{B}_{b}, \tilde{K}_{e}}[\tilde{q};\tilde{K}_b].Pstr_{|\tilde{K}_{ab}|}[\tilde{q}_a;\tilde{x}].\\
&&(\iif\ \tilde{K}_{ab}=\tilde{K}_{ba} \ \then\ key_e!\tilde{K}_{e}.skey!RemStr(\tilde{K}_{ab}, \tilde{x}).\nil\\
&&\hspace{6em}\ \eelse\ (\iif\ \tilde{K}^{\tilde{x}}_{ab}\neq \tilde{K}^{\tilde{x}}_{ba} \ \then\ alarm_a!0.\nil\| alarm_b!0.\nil\\\
&&\hspace{9em} \eelse\ fail!0.\nil))
\end{eqnarray*} 
where to ease the notations, we let $\tilde{K}_{ab}=cmp(\tilde{K}_a, \tilde{B}_a, \tilde{B}_b)$, $\tilde{K}_{ba}=cmp(\tilde{K}_b, \tilde{B}_a, \tilde{B}_b)$, $\tilde{K}^{\tilde{x}}_{ab}=SubStr(\tilde{K}_{ab}, \tilde{x})$, and $\tilde{K}^{\tilde{x}}_{ba}=SubStr(\tilde{K}_{ba}, \tilde{x})$. Similar to $Ran$, the special measurement $Ran'$ here, which takes three parameters, delivers a string of $n$ bits. For example, $Ran_{\tilde{B}_{a}, \tilde{B}_{e}, \tilde{K}_{a}}[\tilde{q}; \tilde{K}_e]$
will first generate a string of $n-|\tilde{K}_{ae}|$ random bits $\tilde{x}$, replace with $\tilde{x}$ the substring of $\tilde{K}_{a}$ at the positions where $\tilde{B}_{a}$ and $\tilde{B}_{e}$ do not match, and store the string after the replacement in $\tilde{K}_{e}$. 
 
%\end{document}

%%%%%%%%%%%%%%%%
\section{Conclusion and related work}
\label{sec:conl}
In our opinion, bisimulations should be considered as a proof methodology for demonstrating behavioural
equivalence between systems, rather than providing the definition of
the extensional behavioural equivalence itself.  We have adapted the
well-known \emph{reduction barbed congruence} used for a variety of
process calculi \cite{ht92,RathkeS08,FournetG05,DH11a}, to obtain a
touchstone extensional behavioural equivalence for quantum processes considered in 
\cite{FDY11}. In the literature there are also minor variations on the
formulation of reduction barbed congruence, often called
\emph{contextual equivalence} or \emph{barbed congruence}.  See \cite{FournetG05,pibook} for a discussion of the
differences.

We have defined a notion of open bisimulations, 
 which
provides both a sound and complete coinductive proof methodology for
establishing the equivalence between qCCS processes.
The operational semantics of this language is given in terms of probabilistic labelled transition systems.
Moreover, we have generalised Hennessy-Milner logic to express the behaviour of quantum processes. In the resulting quantum logic, logical equivalence coincides with open bisimilarity.

To conclude this paper, we would like to compare the open bisimulation
defined here with other bisimulations for quantum processes already
proposed in the literature.
Jorrand and Lalire \cite{JL04, La06} defined a branching bisimulation for their QPAlg, which identifies quantum processes whose associated
graphs have the same branching structure. However, their bisimulation 
cannot always distinguish different quantum operations, as quantum states are only compared when they are input or output. More seriously, the derived bisimilarity is not a congruence; it is not preserved by restriction.
Bisimulation defined in \cite{FDJY07} indeed distinguishes different quantum operations but it
works well only for finite processes, since quantum states are compared after all actions have been performed. Again, it is not preserved by restriction, and whether it is preserved by parallel composition still remains open, although the positive answer is affirmed in two special cases. In \cite{YFDJ09}, a congruent (strong) bisimulation was proposed for a special model where no classical datum is involved. However, as many important quantum communication protocols such as superdense coding and teleportation cannot be described in that model, the scope of its application is very limited. Furthermore, as all quantum operations are regarded as visible in \cite{YFDJ09}, the bisimulation is too strong; it distinguishes two different sequences of quantum operations even when they have the same effect as a whole.

The first general (works for general models where both classical and quantum data are involved, and recursive definition is allowed), weak (quantum operations are regarded as invisible, so that they can be combined arbitrarily), and congruent bisimulation for quantum processes was defined in~\cite{FDY11}. It differentiates quantum input, to match which an arbitrarily chosen super-operator should be considered, from other actions. 
The open bisimulation in this paper makes a step further by treating the super-operator application in an `open' style: applying super-operators before an action to be matched is selected. This makes it possible to separate ground bisimulation and the closedness under super-operator application, and by doing so, we are able to provide not only a neater and simpler definition, but also a powerful technique for proving bisimilarity.

It is easy to prove that the bisimulation in~\cite{FDY11} is both a ground bisimulation and closed under super-operator application. Then by  
Proposition~\ref{prop:ground}, it is also an open bisimulation; in other words, the bisimilarity presented in the current paper is coarser than that defined in~\cite{FDY11}. Whether or not they are actually the same is an interesting question, and we leave it for further investigation.
\bibliographystyle{plain}
\bibliography{ref}

\end{document}
%%%%%%%%%%%%%%%%%%%%%